\documentclass[11pt]{article}

\usepackage{geometry}
\usepackage[numbers]{natbib}
\usepackage{amsfonts,amsmath,amssymb,amsthm}
\usepackage{graphicx}
\usepackage[colorlinks,citecolor=green,linkcolor=blue,bookmarks=false,hypertexnames=true]{hyperref}
\usepackage{mathrsfs}
\usepackage{enumitem}

\usepackage[finalnew]{trackchanges} %[inline] N.B. Hacer \mbox{cite{...}}, \mbox{eqref{...}}.
\addeditor{RA}

\usepackage{titling}
\usepackage{array}
\preauthor{\begin{center}
    \large \lineskip .75em%
        \begin{tabular}[t]{>{\centering\arraybackslash}p{.45\textwidth}}}
        \postauthor{\end{tabular}\par\end{center}}
\makeatletter
\renewcommand\and{% \begin{tabular}
  \end{tabular}%
  \hfill
  \begin{tabular}[t]{>{\centering\arraybackslash}p{.45\textwidth}}}% \end{tabular}
\makeatother

\newtheorem{theorem}{Theorem}
\newtheorem{remark}{Remark}
\newtheorem{definition}{Definition}
\newtheorem{lemma}{Lemma}
\newtheorem{assumption}{Assumption}

\def\const{\text{const}}

\setlist{noitemsep}
\allowdisplaybreaks
\begin{document}

\title{Properties and baroclinic instability of stratified thermal upper-ocean flow}

\author{F.J.\ Beron-Vera\\ Department of Atmospheric Sciences\\ Rosenstiel School of Marine, Atmospheric \& Earth Science\\ University of Miami\\ Miami, Florida, USA\\ \href{mailto:fberon@miami.edu}{\texttt{fberon@miami.edu}} 
\and
M.J.\ Olascoaga\\ Department of Ocean Sciences\\ Rosenstiel School of Marine, Atmospheric \& Earth
Science\\ University of Miami\\ Miami, Florida, USA\\ \href{mailto:jolascoaga@miami.edu}{\texttt{jolascoaga@miami.edu}}} 
\date{Started: October 27, 2023. This version: \today.\\ To appear in \emph{Revista Mexicana de F\'isica}.\vspace{-0.25in}}
\maketitle

\noindent \emph{On the occasion of the 50th anniversary of Departamento de Oceanograf\'ia F\'isica of CICESE (Ensenada, Baja California, Mexico), we dedicate this paper to the memory of Pedro Ripa (1946--2001), who inspired us to follow him in listening to the whispers in the woods.}

\tableofcontents

\begin{abstract}
    We study the properties of, and investigate the stability of a baroclinic zonal current in, a thermal rotating shallow-water model, sometimes called \emph{Ripa's model}, featuring stratification for quasigeostrophic upper-ocean dynamics.  The model has Lie--Poisson Hamiltonian structure.  In addition to Casimirs, the model supports weak Casimirs forming the kernel of the Lie–Poisson bracket for the potential vorticity evolution independent of the details of the buoyancy as this is advected under the flow. The model sustains Rossby waves and a neutral model, whose spurious growth is prevented by a positive-definite integral, quadratic on the deviation from the motionless state.  A baroclinic zonal jet with vertical curvature is found to be spectrally stable for specific configurations of the gradients of layer thickness, vertically averaged buoyancy, and buoyancy frequency.  Only a subset of such states was found Lyapunov stable using the available integrals, except the weak Casimirs, whose role in constraining stratified thermal flow remains to be understood.  The existence of Lyapunov-stable states enabled us to \emph{a priori} bound the nonlinear growth of perturbations to spectrally unstable states. Our results do not support the generality of earlier numerical evidence on the suppression of submesoscale wave activity as a result of the inclusion of stratification in thermal shallow-water theory, which we supported with direct numerical simulations.\\
    
\end{abstract}

\section{Introduction}

The uppermost layer of the ocean, encompassing the region above the main thermocline, including the mixed layer, is marked by the widespread presence of highly energetic mesoscale eddies (20–300 km) with lifespans ranging from weeks to months. These eddies exert a dominant influence on horizontal transport and significantly contribute to vertical transport, thereby impacting the overall biogeochemical cycle \cite[e.g.,][and references therein]{McGillicuddy-16}.

Recent advancements in high-resolution numerical simulations and field observations have shed light on the role of submesoscale motions (1–10 km) in vertical transport dynamics. These motions manifest as Kelvin–Helmholtz-like vortices that form along density fronts and play a significant role in shaping vertical transport processes \cite[e.g.,][and references therein]{Mahadevan-16}.

Mesoscale motions are well described by quasigeostrophic (QG) theory, valid, roughly speaking, for length scales of the order of the gravest baroclinic (i.e., internal) Rossby radius of deformation and time scales much longer than the inverse of the Coriolis parameter \cite{Pedlosky-87}.  The simplest QG model, capable to provide a minimal description of mesoscale upper-ocean flow, has a layer of constant density, limited from above by a rigid lid and from below by a soft interface with an infinitely deep layer of quiescent fluid. Using notation introduced by \cite{Ripa-JFM-95}, we will refer to this model as the \emph{HLQG} for including a \emph{h}omogeneous (density) \emph{l}ayer.  Submesoscale motions, with smaller length scales and faster time scales than mesoscale motions, are, in principle, beyond the scope of QG theory \cite{McWilliams-16}.  However, based on numerical simulations, an interesting observation has been recently made by \cite{Holm-etal-21}, who noted that submesoscale motions can indeed be represented by a QG model with one layer when lateral inhomogeneity is allowed in the buoyancy (temperature) field.  Such simulations revealed sub-deformation scale (i.e., submesoscale) Kelvin--Helmholtz-like vortex rolls reminiscent of structures commonly observed in surface ocean color pictures acquired by satellites.  A result of a simulation of that type along with an outstanding cloud-free ocean color image revealing submesoscale vortices are shown in Fig.\ \ref{fig:color}. 

The above kind of QG model, referred to as \emph{thermal QG model} by \cite{Warneford-Dellar-13}, has a long history.  It was derived by \cite{Ripa-RMF-96} asymptotically from the thermal primitive equations (PE) upon expanding them in small Rossby number.  The thermal PE, or rotating shallow-water \cite{Warneford-Dellar-13}, model was introduced in the 1960s \cite{Obrien-Reid-67}, widely employed from the early 1980s through the early 2000s \cite[e.g.,][]{Schopf-Cane-83, McCreary-Kundu-88, Fukamachi-etal-95, Beier-97, Ochoa-etal-98, Pinet-Pavia-00}, generalized to a system of multiple layers by \cite{Ripa-GAFD-93}, and extensively investigated theoretically by the same author \cite{Ripa-JGR-96, Ripa-JFM-95, Ripa-DAO-99}. Due this author's several contributions, the thermal QG and PE models have been referred to as \emph{Ripa's models} \cite{Dellar-03}.  Honoring notation introduced by \cite{Ripa-JFM-95}, we will refer to these models as models of the \emph{IL$^0$} class, where ``IL'' emphasizes that the models have one (or possibly many) \emph{i}nhomogeneous \emph{l}ayer(s) and the superscript indicates that their fields do not vary in the vertical.  The left panel of Figure \ref{fig:models} shows a cartoon of the vertical structure in the IL$^0$ model class, assuming a reduced-gravity setting.

\begin{figure}
   \centering
   \includegraphics[width=.7\textwidth]{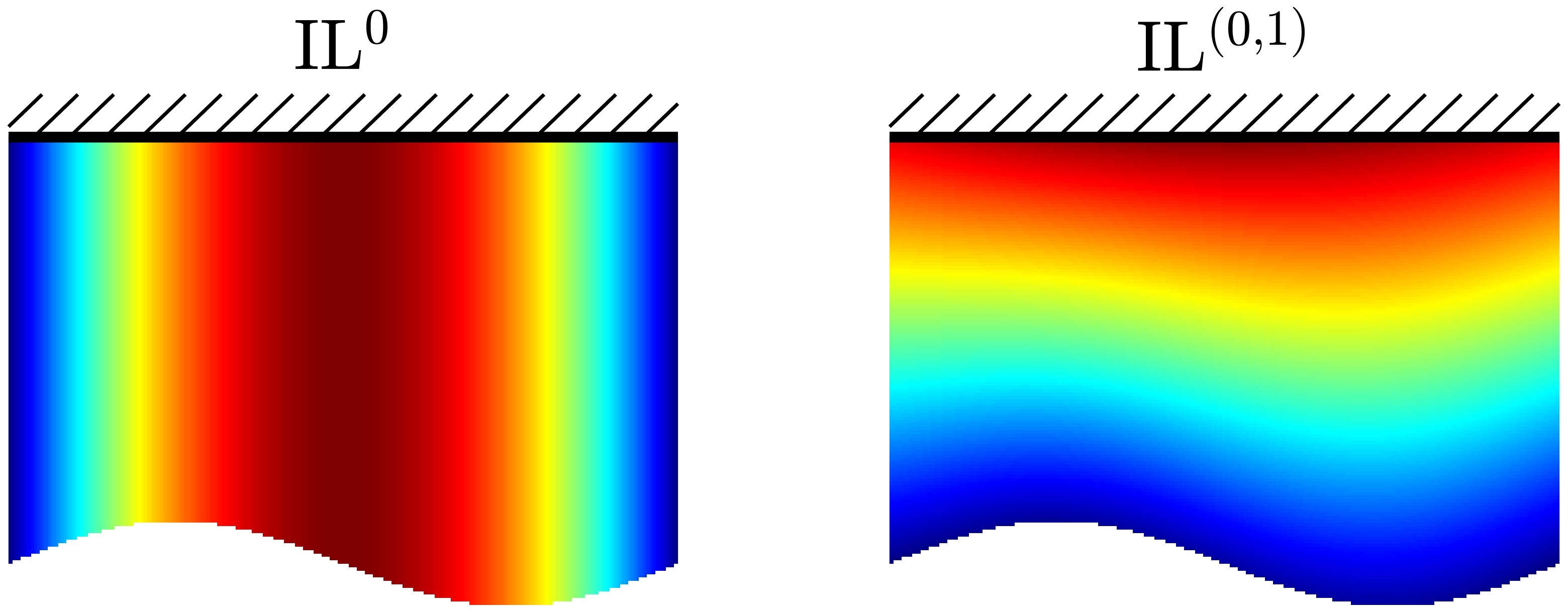}
   \caption{Cartoons depicting the density structure on a vertical plane in the thermal shallow-water model, known as Ripa's model or IL$^0$ (left), and its stratified version analyzed in this paper, denoted as IL$\smash{^{(0,1)}}$ (right). A reduced-gravity setting is assumed, featuring one active layer bounded from above by a rigid lid and from below by a soft interface, which interfaces with an inert layer.}
   \label{fig:models}
\end{figure}

In the last decade, the IL$^0$ has made a strong comeback \cite{Warneford-Dellar-14, Gouzien-etal-17, Eldred-etal-19, Lahaye-etal-20, Kurganov-etal-20, Holm-etal-22, Andrade-etal-22, Cao-etal-23, Crisan-etal-23, Lahaye-etal-24}, motivated in large part by the tendency of lateral temperature gradients in the upper ocean to become steeper as the ocean absorbs heat from a warming troposphere.  In particular, \cite{Beron-21-POFb} introduced a family of IL-type models, both in PE and for QG dynamics, featuring stratification.  The buoyancy field in this model family, referred to as IL$^{(0,n)}$, is written as a polynomial in the vertical coordinate of arbitrary whole degree ($n$) with coefficients that vary laterally and with time while they are advected (by Lie transport) under the flow of the model horizontal velocity, which is vertically uniform, assigning meaning to the first slot in the superscript in IL$^{(0,n)}$.  The IL$^{(0,n)}$PE models admit Euler--Poincar\'e variational formulation \cite{Holm-etal-02} and, via a partial Legendre transform, possess a Lie--Poisson, i.e., noncanonical, Hamiltonian structure \cite{Morrison-Greene-80}.  The IL$^{(0,n)}$QG, derived asymptotically from the IL$^{(0,n)}$PE, are Hamiltonian in a Lie--Poisson sense. Including stratification in the IL models expands their scope, for instance, by enabling them to represent important processes such as the restratification of the mixed layer, which follows the development of submesoscale motions by baroclinic instability \cite{Boccaletti-etal-07}.  (In an attempt to add more physics to the IL models, velocity vertical shear is included in \cite{Ripa-JFM-95} in the single layer case and in \cite{Beron-21-RMF} in a system with multiple layers. However, the resulting IL models with stratification and vertical shear do not enjoy the aforementioned geometric mechanics properties.)

\subsection{Goals of the paper}

In this paper we investigate in some depth the properties of the recently proposed stratified IL models by focusing on the IL$^{(0,1)}$QG (Figure \ref{fig:models}, right panel). We also study how the model fares with respect to baroclinic instability, which, as noted, is an important source of submesoscale circulatory motions in the ocean mixed layer. 

It may sound awkward to attempt an investigation of baroclinic instability using an IL model involving a single layer since the (horizontal) velocity is vertically uniform. Two layers may seem to be required as is needed to represent baroclinic instability using the HLQG in the so-called Phillips problem \cite{Pedlosky-87}. However, by the thermal-wind balance, which dominates at low frequency, in (every layer of) an IL model the velocity field \emph{implicitly} includes shear in the vertical.  Furthermore, implicitly, the velocity vertical profile in the IL$^{(0,1)}$ is curved.  This extends the standard baroclinic instability problem of Phillips, or Eady in the continuously stratified case \cite{Pedlosky-87}, from uniform vertical shear (linear vertical profile) to velocity with vertical curvature, which augments the dimension of the space of parameters for exploration.

While the study of this paper is of interest \emph{per se}, a motivation for pursuing it is found in numerical simulations presented by \cite{Beron-21-POFb}. These simulations have suggested that inclusion of stratification in the IL models tends to halt the development of submesoscale circulations.  \change[RA]{This followed from the comparison of evolutions in the IL$^0$QG and IL$^{(0,1)}$QG forced by bottom topography, which can be included without spoiling the Hamiltonian structure of the models. (This was shown in previous work \mbox{\cite{Beron-21-POFa, Holm-etal-21}} to hold for the IL$^0$QG and here we show that it also applies to the IL$^{(0,1)}$QG.) }{This followed from the comparison of forced evolutions in the IL$^0$QG and IL$^{(0,1)}$QG. The forcing, which can be interpreted as a topographic forcing \mbox{\cite{Holm-etal-21}}, does not spoil the Hamiltonian structure of the models \mbox{\cite{Beron-21-POFa, Holm-etal-21}}.} The study carried out here will help assess the generality of the numerical evidence presented in \cite{Beron-21-POFb}.

\begin{figure}
    \centering
    \includegraphics[width=.35\textwidth]{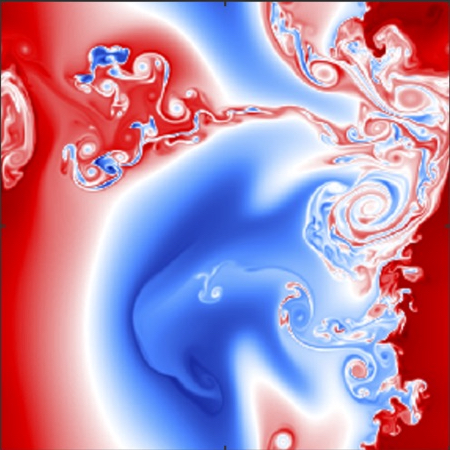}\,
    \includegraphics[width=.35\textwidth]{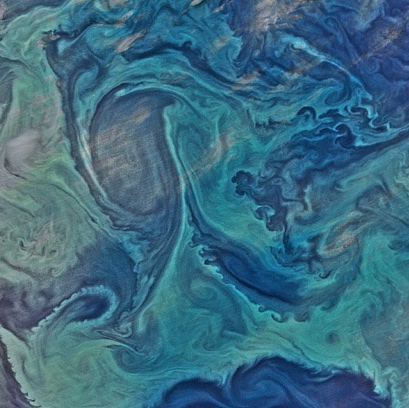}
    \caption{(left) Emergent cascade of submesoscale vorticity filament rollups in a reduced-gravity direct numerical simulation of the IL$^0$QG with ``topographic'' forcing in a doubly periodic domain $\mathbb R/R\mathbb Z \times \mathbb R/R\mathbb Z$ of the $\beta$-plane, where $R \approx 25$ km is the baroclinic Rossby radius of deformation. (right) Ocean color image acquired by VIIRS (Visible Infrared Imaging Radiometer Suite) on 1 January 2015 west of the Drake Passage in the Southern Ocean, revealing vortices with diameters ranging from a couple of km to a couple of hundred km.  Image credit: NASA Ocean Color Web (\href{https://oceancolor.gsfc.nasa.gov/gallery/447/}{https://oceancolor.gsfc.nasa.gov/gallery/447/}).}
    \label{fig:color}
\end{figure}

\subsection{Paper organization}

The rest of the paper is organized as follows.  In Section \ref{sec:model} we review the IL$^{(0,1)}$QG and describe its properties.  Specifically, Section \ref{sec:eqs} presents the model equations along with its boundary conditions and a physical interpretation of its variables. Section \ref{sec:uniqueness} discusses uniqueness of solutions in the model.  Section \ref{sec:circulations} is dedicated to discussing theorems of circulation, of an appropriately defined Kelvin--Noether-like quantity along material loops (Section \ref{sec:Kelvin}) and of the velocity along solid boundaries of the flow domain (Section \ref{sec:coasts}).  Invariant subspaces of the model are considered in Section \ref{sec:invariant}. In Section \ref{sec:hamiltonian} the Lie--Poisson Hamiltonian structure is clarified, particularly in relation with boundary conditions, which were only superficially discussed in \cite{Beron-21-POFb}. Inclusion of topographic forcing without destroying Hamiltonian structure is treated in Section \ref{sec:topo}. Reviewing the Hamiltonian structure of the IL$^{(0,1)}$QG is helpful for Section \ref{sec:integrals}, which is devoted to describing the conservation laws of the system, which, when linked to explicit symmetries, are obtained via a Noether's theorem appropriate for noncanonical Hamiltonian systems.  The conservation laws include a family of integrals of motion that form the kernel of the Lie–Poisson bracket for the potential vorticity evolution independent of buoyancy details as the buoyancy is transported by the flow.  Linear waves are discussed in Section \ref{sec:waves}.  Section \ref{sec:stability} is dedicated to investigate the stability of a baroclinic zonal flow on the $\beta$-plane with vertical curvature.  Three types of stability notions are discussed, namely, spectral (Section \ref{sec:spectral}), formal (Section \ref{sec:formal}), and Lyapunov (Section \ref{sec:lyapunov}).  In Section \ref{sec:bounds} we derive a-priori bounds on the nonlinear growth of perturbations to unstable basic states.  A discussion, accompanied by direct numerical simulations, of differences between behavior in the IL$^{(0,1)}$QG and IL$^0$QG is offered in Section \ref{sec:discussion}.  Section \ref{sec:conclusions} presents the conclusions from the paper.  \add[RA]{A glossary containing relevant notions from geophysical fluid dynamics is included following this section.} Finally, an Appendix is included containing numerical details of the direct simulations.

\section{The IL$^{(0,1)}$QG}\label{sec:model}

Let $\mathbf x = (x,y)$ denote position on a periodic zonal domain $\mathscr D = \mathbb R/L\mathbb Z \times [0,W]$ of the $\beta$ plane.  That is, the channel rotates steadily at angular velocity $\tfrac{1}{2}f\hat{\mathbf z}$, where $f = f_0 + \beta y$ is the Coriolis parameter.  We will denote by $\partial \mathscr D_-$ (resp., $\partial \mathscr D_+$) the southern (resp., northern) wall of the channel at $y = 0$ (resp., $y = W$). Consider a layer of (ideal) fluid limited above by a horizontally rigid lid at $z = 0$ and below by a soft interface at $z = -h(\mathbf x,t)$ with an infinitely deep layer of homogeneous fluid at rest. As in (rotating) shallow-water theory, the motion in the active (top) layer is assumed to be columnar. The vertically shearless horizontal velocity is denoted by $\bar{\mathbf u}(\mathbf x,t)$, where the overbar represents a vertical average. Similar to conventional \emph{thermal} shallow-water theory, the buoyancy of the active layer is permitted to change both horizontally and temporally \cite{Ripa-GAFD-93}. However, unlike this theory, vertical variations in buoyancy are also permitted \cite{Beron-21-POFb}.  In the IL$^{(0,1)}$ class of \emph{stratified} thermal shallow-water models, of interest here, this is done by writing the buoyancy as 
\begin{equation}
    \vartheta(\mathbf x,z,t) = \bar{\vartheta}(\mathbf x,t) + \left(1 + \frac{2z}{h(\mathbf x,t)}\right) \vartheta_\sigma(\mathbf x,t).
    \label{eq:thvarphi}
\end{equation}
That is, the stratification in the active layer is set to be be uniform, with the buoyancy (linearly) varying from $\bar\vartheta + \vartheta_\sigma$ at $z = 0$ to $\bar\vartheta - \vartheta_\sigma$ at $z = -h$.  This imposes the condition
\begin{equation}
    \bar\vartheta \ge \vartheta_\sigma > 0,
    \label{eq:vertstab}
\end{equation}
stating that the buoyancy is everywhere positive and the density increases with depth.  For future reference, we introduce the instantaneous buoyancy (or Brunt--V\"ais\"al\"a) frequency squared
\begin{equation}
    n^2(\mathbf x,z,t) := \partial_z\vartheta(\mathbf x,z,t) = 2\frac{\vartheta_\sigma(\mathbf x,t)}{h(\mathbf x,t)}.
\end{equation}

\begin{remark}[On notation]
  The subscript notation is motivated by denoting by $\sigma$, as in \textup{\cite{Ripa-JFM-95}}, the scaled vertical coordinate $1 + \smash{\frac{2z}{h(\mathbf x,t)}}$.  For instance, in the IL$^{(0,2)}$ model class \textup{\cite{Beron-21-POFb}} the buoyancy is written as $\vartheta = \bar\vartheta + \sigma\vartheta_\sigma + (\sigma^2-\frac{1}{3})\vartheta_{\sigma^2}$, clarifying additional notation introduced below.
\end{remark}
 
\subsection{Model equations}\label{sec:eqs}

The evolution equations of the IL$^{(0,1)}$QG in the reduced-gravity setting above are given by \cite{Beron-21-POFb}
\begin{subequations}
\begin{gather}
    \partial_t\bar\xi + \smash{\big\{\bar\psi,\bar\xi - R_S^{-2} (\psi_\sigma-\tfrac{2}{3}\psi_{\sigma^2})\big\}_{xy}} = 0,\\
    \partial_t\psi_{\sigma} + \{\bar\psi,\psi_{\sigma}\}_{xy} = 0,\\
    \partial_t\psi_{\sigma^2} + \{\bar\psi,\psi_{\sigma^2}\}_{xy} = 0,
\end{gather}
where
\begin{align}
    \nabla^2\bar\psi - R_S^{-2}\bar\psi =  \bar\xi - R_S^{-2}(\psi_\sigma-\tfrac{2}{3}\psi_{\sigma^2}) - \beta y
    \label{eq:inv}
\end{align}
with
\begin{equation}
    R_S^2 := \left(1 - \tfrac{1}{3}S\right) R^2.    
\end{equation}  
Here
\begin{equation}
   \{a,b\}_{xy} := \nabla^\perp a\cdot \nabla b
   \label{eq:canonical}
\end{equation}
for any differentiable functions $a,b(\mathbf x,t)$, where $\nabla^\perp$ is short for $\hat{\mathbf z}\times\nabla$, is the Jacobian of the transformation $(x,y)\mapsto (a,b)$ and 
\begin{equation}
    R := \frac{\sqrt{g_\mathrm{r}H_\mathrm{r}}}{|f_0|},\quad 
    S := \frac{N_\mathrm{r}^2H_\mathrm{r}}{2g_\mathrm{r}}.
    \label{eq:RS}
\end{equation}
\label{eq:IL01}%
\end{subequations}
The positive constants $H_\mathrm{r}$, $g_\mathrm{r}$, and $N_\mathrm{r}^2$ are reference (i.e., in the absence of currents) layer thickness, vertically averaged buoyancy,\footnote{In \cite{Beron-21-POFb}, $2g_\mathrm{r}$ is incorrectly referred to as the reference buoyancy jump at the bottom of the active layer.} and buoyancy frequency squared, respectively. The (positive) constant on the left in \eqref{eq:RS} corresponds to the equivalent barotropic (external) Rossby radius of deformation in a model with arbitrary stratification in a reduced-gravity setting.  The constant on the right in \eqref{eq:RS} is a measure of reference stratification.  The reference buoyancy varies from $1 + S$ at $z = 0$ to $1 - S$ at $z = -H_\mathrm{r}$. Thus, by static stability, $0 < S < 1$.

The system is subjected to
\begin{subequations}
\begin{equation}
    \nabla^\perp\bar\psi\cdot\hat{\mathbf n}\vert_{\partial \mathscr D_\pm} = 0 \Longleftrightarrow \partial_x\vert_{y=0,W}\bar\psi = 0
\end{equation}
and periodicity in the along channel direction, $x$, i.e.,
\begin{equation}
    \bar\xi(x+L,y,t) = \bar\xi(x,y,t)
\end{equation}
\label{eq:IL01-BCs}%
\end{subequations}
and similarly for the rest of the variables.

The IL$^{(0,1)}$QG has three prognostic fields, $(\bar\xi,\psi_\sigma, \psi_{\sigma^2})$, assumed to be sufficiently smooth in each of its arguments, $(\mathbf x,t)$. These diagnose $\bar\psi(\mathbf x,t)$ via \eqref{eq:inv}, which defines the invertibility principle for the IL$^{(0,1)}$QG.  (In Fourier space, say, the positive-definite symmetric operator $R_S^{-2} - \nabla^2$ is nondegenerate; hence, its operator inverse $(R_S^{-2} - \nabla^2)^{-1}$ exists and is well defined.) 

The physical meaning of the above variables is as follows. Let $\mathrm{Ro} > 0$ be a small parameter taken to represent a Rossby number, measuring the strength of inertial and Coriolis forces, e.g., 
\begin{equation}
  \mathrm{Ro} = \frac{V}{|f_0|R} \ll 1, 
\end{equation}
where $V$ is a characteristic velocity.  In QG theory \cite[e.g.,][]{Pedlosky-87} it is assumed that 
\begin{equation}
  (|\bar{\mathbf u}|,h-H_\mathrm{r},\partial_t,\beta y) = O(\mathrm{Ro}\hspace{0.08334em}
  V, \mathrm{Ro}\hspace{0.08334em} R, \mathrm{Ro} f_0,  \mathrm{Ro} f_0).
  \label{eq:QGscaling}
\end{equation}
Consistent with the QG scaling \eqref{eq:QGscaling}, with an $O(\mathrm{Ro}^2)$ error, one has:
\begin{gather}
	 \bar{\mathbf u} =  \nabla^\perp\bar\psi,\label{eq:u}\\ 
	 h = H_\mathrm{r} + \frac{H_\mathrm{r}}{f_0R_S^2}(\bar\psi - \psi_\sigma + \tfrac{2}{3}\psi_{\sigma^2})\label{eq:h}\\
  \bar\vartheta = g_\mathrm{r} + \frac{2g_\mathrm{r}}{f_0R^2}\psi_\sigma,\label{eq:vartheta}\\
	 \vartheta_\sigma = \tfrac{1}{2}N^2_\mathrm{r}H_\mathrm{r} + \frac{4g_\mathrm{r}}{f_0R^2}\psi_{\sigma^2},\label{eq:vartheta_sigma}
\end{gather}
where
\begin{equation}
    \frac{\psi_\sigma - 2\psi_{\sigma^2}}{f_0R^2} \ge 
    \frac{S - 1}{2},\quad 
    \frac{\psi_{\sigma^2}}{f_0R^2} > -\frac{S}{4},
\end{equation}
by \eqref{eq:vertstab}.  Finally, with an $O(\mathrm{Ro}^2)$ error,  
\begin{equation}
    \frac{\nabla^\perp\cdot\bar{\mathbf u} + f}{h} = \frac{f_0 + \bar\xi}{H_\mathrm{r}}.
    \label{eq:q}
\end{equation}
The left-hand-side of \eqref{eq:q} is proportional to the Ertel $\frac{z}{h}$-potential vorticity in the hydrostatic approximation and with the horizontal velocity replaced with $\bar{\mathbf u}$ \cite[cf.][for details]{Ripa-JFM-95}. Thus $\bar\xi$ can be identified with potential vorticity in the IL$^{(0,1)}$QG.

With the identifications \eqref{eq:u}--\eqref{eq:q}, equation (\ref{eq:IL01}a) controls the evolution of IL$^{(0,1)}$QG potential vorticity. This quantity is not materially conserved, i.e., advected (by Lie transport) under the flow of $\nabla^\perp\bar\psi$. Rather, the potential vorticity in the IL$^{(0,1)}$QG is created (or annihilated) by the misalignment between the gradients of layer thickness and buoyancy (cf.\ Section \ref{sec:Kelvin}, for additional discussion). This is consistent with Ertel's $\frac{z}{h}$-potential vorticity not being materially conserved.  Equations (\ref{eq:IL01}b) and (\ref{eq:IL01}c) control, respectively, the evolution of the vertical average and gradient of buoyancy.  These are both materially conserved.  Altogether (\ref{eq:IL01}b) and (\ref{eq:IL01}c) give material conservation of buoyancy in the IL$^{(0,1)}$QG.  Finally, the boundary condition (\ref{eq:IL01-BCs}a) means no flow through the solid walls of the zonal channel domain.

\begin{remark}[Implicit vertical shear]
    By the thermal-wind balance, the velocity has implicit vertical curvature, which motivates the streamfunction notations $\psi_\sigma$ and $\psi_{\sigma^2}$ \textup{\cite{Beron-21-POFb}}.  Indeed, the buoyancy distribution \eqref{eq:thvarphi} implicitly implies that the velocity is determined, with an $O(\mathrm{Ro}^2)$ error, by the streamfunction 
    \begin{equation}
       \psi = \bar\psi +  \left(1 + \frac{2z}{H_\mathrm{r}}\right) \psi_\sigma + \left(\left(1 + \frac{2z}{H_\mathrm{r}}\right)^2 - \frac{1}{3}\right) \psi_{\sigma^2}.
    \end{equation}
\end{remark}

We note that the IL$^0$QG can be recovered from system \eqref{eq:C-IL01} upon omitting $\psi_{\sigma^2}$ and setting $S = 0$, i.e., replacing $R_S$ with $R$. Note, for instance, that the buoyancy restriction in the IL$^0$QG is $\psi_\sigma > - \smash{\tfrac{1}{2}}f_0R^2$.  The potential vorticity in the IL$^0$QG, given by $\bar\xi = \nabla^2\bar\psi - R^{-2}(\bar\psi - \psi_\sigma) + \beta y$, is created (or annihilated) as in the IL$^{(0,1)}$QG due to lateral changes in buoyancy, which is only permitted to vary in the horizontal (and time).  This is unlike the HLQG, which follows from the IL$^0$QG after ignoring $\psi_\sigma$.  The potential vorticity in the HLQG, $\bar\xi = \nabla^2\bar\psi - R^{-2}\bar\psi + \beta y$, is materially conserved. The results discussed below for the IL$^{(0,1)}$QG can be translated to corresponding results for the IL$^0$QG or HLQG in the indicated limits.

\subsection{Uniqueness of solutions}\label{sec:uniqueness}

Assume that $\bar\psi_1(\mathbf x,t)$ and $\bar\psi_2(\mathbf x,t)$ are two solutions of \eqref{eq:inv}, on the zonal channel domain $\mathscr D$, each one satisfying boundary conditions \eqref{eq:IL01-BCs}. Let $\phi := \bar\psi_1 - \bar\psi_2$; then
\begin{equation}
    \nabla^2\phi - R_S^{-2}\phi = 0 
    \label{eq:psi}
\end{equation}
on $\mathscr D$ subjected to
\begin{equation}
    \nabla^\perp\phi\cdot\hat{\mathbf n}\vert_{\partial \mathscr D_\pm} = 0 \Longleftrightarrow \partial_x\vert_{y=0,W} \phi= 0,\quad \phi(x,0,t) = \phi(x,L,t).
    \label{eq:psi-BCS}
\end{equation}
Multiplying \eqref{eq:psi} by $\phi$ and integrating by parts with \eqref{eq:psi-BCS} in mind,
\begin{equation}
    \int_{\mathscr D} |\nabla\phi|^2 + R_S^{-2}\phi^2\,d^2x = 0,
\end{equation}
which holds if and only if $\phi = 0$, i.e., $\psi_1 = \psi_2$.

\subsection{Circulations}\label{sec:circulations}

\subsubsection{A Kelvin circulation theorem}\label{sec:Kelvin}

Let $\mathbf f(\mathbf x)$ be such that $\nabla^\perp\cdot\mathbf f = \beta y$.  Let in addition $\partial \mathscr D_t$ be the material loop enclosing a material region $\mathscr D_t$ at time $t>0$ after being advected from its initial position at time $t=0$ by the flow of $\nabla^\perp\bar\psi$. Consider 
\begin{equation}
    \mathcal K := \oint_{\partial\mathscr D_t}\big(\nabla^\perp\bar\psi + \textbf{f} - R_S^{-2}\nabla^{-2}\nabla^\perp(\bar\psi - \psi_\sigma + \tfrac{2}{3}\psi_{\sigma^2})\big) \cdot d\textbf{x} = \int_{\mathscr D_t}\bar\xi\,d^2x,
    \label{eq:K}
\end{equation}
where the last equality follows upon using Stokes' theorem.  The above is an appropriate definition of the \emph{Kelvin circulation for the IL$^{(0,1)}$QG} as it leads to a theorem analogous to the Kelvin--Noether theorem for the IL$^{(0,1)}$PE \cite{Beron-21-POFb}. It extends that one found in Section 3.1 of \cite{Holm-15b} for adiabatic (i.e., HL) QG dynamics to stratified thermal (i.e., IL$^{(0,1)}$) QG dynamics. By (\ref{eq:IL01}a) and after pulling back the integrand in the last equality of \eqref{eq:K} to $t=0$ and changing variables, it follows that
\begin{equation}
    \dot{\mathcal K} = R_S^{-2}\int_{\mathscr D_t}\big\{\bar\psi,\psi_\sigma-\tfrac{2}{3}\psi_{\sigma^2}\big\}_{xy}\,d^2x.
    \label{eq:dotK}
\end{equation}
As in the IL$^{(0,1)}$PE, the misalignment between the gradients of layer thickness and buoyancy creates (or destroys) Kelvin circulation. This is most cleanly seen by noting that
\begin{equation}
    \int_{\mathscr D_t}\nabla^\perp h \cdot \nabla\left(\bar\vartheta - \tfrac{1}{3}\vartheta_\sigma\right)\,d^2x = 2R_S^{-2}\int_{\mathscr D_t}\big\{\bar\psi,\psi_\sigma-\tfrac{2}{3}\psi_{\sigma^2}\big\}_{xy}\,d^2x   
\end{equation}
with an $O(\mathrm{Ro}^3)$ error.
\begin{remark}
Two exceptions are\change[RA]{as follows:}{the cases when $\partial \mathscr D_t$ is isopycnic at each level in the vertical and when this is solid, as is the case of the boundary of the zonal channel domain $\mathscr D$. In each of these situations, the Kelvin circulation is conserved. The former follows by noting that the right-hand side of \mbox{\eqref{eq:dotK}} equals $-\smash{\oint_{\partial \mathscr{D}_t}} \bar{\psi} \nabla^\perp (\psi_\sigma - \tfrac{2}{3} \psi_{\sigma^2}) \cdot \hat{\mathbf{n}} ds$, while the latter follows by realizing that it is also equal to $\smash{\oint_{\partial \mathscr{D}_t}} (\psi_\sigma - \tfrac{2}{3} \psi_{\sigma^2}) \nabla^\perp \bar{\psi} \cdot \hat{\mathbf{n}} ds$. Kelvin circulation conservation, when $\partial \mathscr{D}_t$ is solid, differs from that in the IL$^{(0,1)}$PE, where it must be levelwise isopycnic, making the IL$^{(0,1)}$QG less restrictive than the IL$^{(0,1)}$PE}

\remove[RA]{Note that $\smash{\{\bar\psi,\psi_\sigma-\tfrac{2}{3}\psi_{\sigma^2}\}_{xy}} = \smash{\nabla^\perp\bar\psi}\cdot \nabla(\psi_\sigma-\tfrac{2}{3}\psi_{\sigma^2}) = - \smash{\nabla\bar\psi\cdot \nabla^\perp}(\psi_\sigma-\tfrac{2}{3}\psi_{\sigma^2})$. Since $\nabla\cdot\nabla^\perp = 0$, then by Stokes' theorem the integral on the right-hand-side of \mbox{\eqref{eq:dotK}} is equal to $-\smash{\oint_{\partial \mathscr D_t}}\bar\psi\nabla^\perp(\psi_\sigma - \tfrac{2}{3}\psi_{\sigma^2})\cdot\hat{\mathbf n}ds$. This vanishes if $\psi_\sigma$ and $\psi_{\sigma^2}$ are constant along $\mathscr D_t$. The Kelvin circulation is conserved in these circumstances, which express that the material loop $\partial \mathscr D_t$ is isopycnic at each level in the vertical.}
    
\remove[RA]{When $\mathscr D_t$ is replaced by the zonal channel domain $\mathscr D$, then the Kelvin circulation is conserved, whether $\partial \mathscr D$ is levelwise isopycnic or not. To see this, first note that, in such a case, the integral on right-hand-side of \mbox{\eqref{eq:dotK}} is replaced by $\smash{\int_{\mathscr D}} \smash{\{\bar\psi,R_S^{-2}(\psi_\sigma-\tfrac{2}{3}\psi_{\sigma^2}) - \bar\xi\}_{xy}}\,d^2x$. By (\mbox{\ref{eq:IL01-BCs}b}) this is equal to $\smash{(\int_{\partial \mathscr D_-} - \int_{\partial \mathscr D_+})}\smash{(R_S^{-2}(\psi_\sigma - \tfrac{2}{3}\psi_{\sigma^2}) - \bar\xi)}\nabla^\perp\bar\psi\cdot\hat{\mathbf n}ds$, upon invoking Stokes' theorem. By (\mbox{\ref{eq:IL01-BCs}b}) each term of this sum individually vanishes. This is unlike the IL$^{(0,1)}$PE, in which case $\partial \mathscr D_\pm$ must be levelwise isopycnic, making the IL$^{(0,1)}$QG less restrictive than the IL$^{(0,1)}$PE. (In that model and the present one, if this condition is imposed initially, then it is preserved for all time.)}
\label{rem:Kelvin}
\end{remark}

\subsubsection{Circulation of the velocity}\label{sec:coasts}

Consider
\begin{equation}
    \gamma_\pm := \int_{\partial \mathscr D_\pm} \nabla\bar\psi\cdot\hat{\mathbf n}ds = \int_0^L\partial_y\vert_{y=0,W}\bar\psi\,dx,
\end{equation}
defining the \emph{velocity circulations} along the coasts of the zonal channel domain.  Note the difference with \eqref{eq:K}, which we have referred to as a Kelvin circulation for the IL$^{(0,1)}$QG.  We will \emph{impose} the condition
\begin{equation}
    \dot\gamma_\pm = 0.
    \label{eq:dotgamma}
\end{equation}
Constancy of $\gamma_\pm$ is a boundary condition, in addition to \eqref{eq:IL01-BCs}, which \emph{guarantees} energy preservation (cf.\ Section \ref{sec:energy-momentum-casimir}). This is not a peculiarity of the IL$^{(0,1)}$QG, but a rather general aspect of QG models.

\remove[RA]{We note that when $\mathscr D$ is replaced by $\mathscr D_t$ in the definition of the Kelvin circulation \mbox{\eqref{eq:K}}, one has the relationship}
\remove[RA]{$\mathcal K = \bar\psi\vert_{\partial \mathscr D_-}\gamma_- + \bar\psi\vert_{\partial \mathscr D_+}\gamma_+ + \int_{\mathscr D} \beta y - R_S^{-2}(\bar\psi - \psi_\sigma + \tfrac{2}{3}\psi_{\sigma^2})\,d^2x,$}
\remove[RA]{as it follows by using \mbox{\eqref{eq:inv}}.  By Remark \mbox{\ref{rem:Kelvin}} and (\mbox{\ref{eq:IL01}}b)--(\mbox{\ref{eq:IL01}}c), it follows that}
\remove[RA]{$\frac{d}{dt}\int_{\mathscr D} \bar\psi\,d^2x = 0.$}
\remove[RA]{This result may be seen to express volume preservation with an $O(\mathrm{Ro}^3)$ error.  Indeed, one may deduce this result by starting from the advection equation of $hd^2x$ or local volume conservation law in the IL$^{(0,1)}$PE, viz., $\partial_t h + \nabla\cdot h\bar{\mathbf u} = 0$. Then one may use \mbox{\eqref{eq:h}} and apply the QG asymptotic expansion implied by the scaling \mbox{\eqref{eq:QGscaling}} noting that, in the IL$^{(0,1)}$PE, $\bar\vartheta$ and $\vartheta_\sigma$ are materially conserved. That is, $\bar\vartheta$ and $\vartheta_\sigma$ represent advected scalar quantities, evolving therefore as $\partial_t\bar\vartheta + \bar{\mathbf u}\cdot\nabla \bar\vartheta = 0$ and $\partial_t\vartheta_\sigma + \bar{\mathbf u}\cdot\nabla \vartheta_\sigma = 0$, respectively. However, nothing of this is available explicitly in the QG limit. In other words, \mbox{\eqref{eq:dotgamma}} cannot be deduced from \mbox{\eqref{eq:vol}} by working with the IL$^{(0,1)}$QG model equations \mbox{\eqref{eq:IL01}}.}

\subsection{Invariant subspaces}\label{sec:invariant}

If system \eqref{eq:IL01} is initialized from $\psi_\sigma, \psi_{\sigma^2} = \const$, then these variables preserve their initial constant values at all times.  In other words, the subspace $\{\psi_\sigma,\psi_{\sigma^2} = \const\}$ is an invariant subspace of the IL$^{(0,1)}$QG. The dynamics on this invariant subspace is formally the same as that in the adiabatic case, that is, the HLQG, wherein the potential vorticity, given by $\bar\xi = \nabla^2\bar\psi - R^{-2}\bar\psi + \beta y$, is materially conserved. This holds formally because $\psi_{\sigma^2}$ actually is a perturbation on a reference uniform stratification.  But this is reflected in \eqref{eq:IL01} only through the stratification parameter $S$. The HLQG and IL$^{(0,1)}$QG potential vorticities on $\{\psi_\sigma,\psi_{\sigma^2} = \const\}$ differ by unimportant constants and by the Rossby radius of deformation being smaller for $S\neq 0$. In turn, if \eqref{eq:IL01} is initialized from $\psi_{\sigma^2} = \const$, then this is preserved for all time. The dynamics on this invariant subspace is formally the same as that of the IL$^0$QG, with the caveats noted above. Similarly, in the IL$^0$QG, obtained from the IL$^{(0,1)}$QG by ignoring $\psi_{\sigma^2}$ and setting $S = 0$, the subspace $\{\psi_\sigma = \const\}$ is invariant and on this invariant subspace the dynamics coincides with that of the HLQG, exactly in this case.

\subsection{Hamiltonian structure}\label{sec:hamiltonian}

We begin with a definition that is helpful to understand operations in the rest of the paper.  This is followed by an assumption, which enables us to frame the Hamiltonian structure of the IL$^{(0,1)}$QG in a domain with boundaries consistent with the aforementioned definition.

\begin{definition}[Functional, functional derivatives, and variations]
A functional $\mathcal F$ of scalar fields $\varphi^j(\mathbf x,t)$, $j = 1,\dotsc,J$, denoted $\mathcal F[\varphi]$, is a map from a vector space $V$ to $\mathbb R$. The functional derivatives of $\mathcal F$ are defined as the unique elements $\smash{\frac{\delta\mathcal F}{\delta \varphi^j}}$, $j = 1,\dotsc,J$, satisfying
\begin{equation}
    \left.\frac{d}{d\varepsilon}\right\vert_{\varepsilon = 0}\mathcal F[\varphi + \varepsilon\delta \varphi] =: \sum_{j=1}^J\int_{\mathscr D} \frac{\delta \mathcal F}{\delta \varphi^j}\delta \varphi^j\,d^2x =: \delta \mathcal F[\delta \varphi;\varphi],
\end{equation}
where $\varepsilon$ is a real parameter and $\delta \varphi$ is called the variation of $\varphi$. The functional $\delta\mathcal F$, linear in $\delta \varphi$, is called the first variation of $\mathcal F$. The above is the Gateaux definition of functional derivative, which does not require $V$ to be normed. However, it is convenient to be taken as such, so contributions of the first and higher-order variations to the total variation can be measured. More specifically, the second variation of $\mathcal F$, quadratic in $\delta \varphi$,
\begin{equation}
    \delta^2\mathcal F[\delta \varphi;\varphi] := \left.\frac{d^2}{d\varepsilon^2}\right\vert_{\varepsilon = 0}\mathcal F[\varphi + \varepsilon\delta \varphi] =: \sum_{j,k=1}^J\int_{\mathscr D} \frac{\delta^2 \mathcal F}{\delta \varphi^j\delta \varphi^k}\delta \varphi^j\delta \varphi^k\,d^2x,    
\end{equation}
where $\smash{\frac{\delta^2\mathcal F}{\delta \varphi^j\delta \varphi^k}}$, $i,k = 1,\dotsc, J$, define the second functional derivatives of $\mathcal F$.  Higher-order variations, as well as functional derivatives, are similarly defined, so that the total variation of $\mathcal F$,
\begin{equation} 
    \Delta \mathcal F[\delta \varphi;\varphi] := \mathcal F[\varphi + \delta \varphi] - \mathcal F[\varphi] = \delta\mathcal F[\delta \varphi;\varphi] + \tfrac{1}{2!}\delta^2\mathcal F[\delta \varphi;\varphi] + \tfrac{1}{3!}\delta^3\mathcal F[\delta \varphi;\varphi] + \cdots,
\end{equation}
where $\delta^n\mathcal F = O(\|\delta\varphi\|^n)$ with $\|\cdot\|^2 = \int_\mathscr{D} (\cdot)^2\,d^2x$ is the square of the $L^2$ norm.
\label{def:fun}
\end{definition}

\begin{assumption}
We will assume that all variations in the sense of Definition \ref{def:fun} are restricted to those that preserve the circulation of the velocity along each wall of the zonal channel domain, namely,
\begin{equation}
    \delta\gamma_\pm = 0.
    \label{eq:dgamma}
\end{equation}
\label{ass:gamma}
\end{assumption}

The restriction \eqref{eq:dgamma} enables a Hamiltonian formulation of the IL$^{(0,1)}$QG wherein variational calculus is consistent with Definition \ref{def:fun}. Otherwise, the space space variables must be augmented to include $\gamma_\pm$ as done in \cite{Holm-etal-85}. However, this requires one to define the notion of variational derivative differently than in Definition \ref{def:fun} \cite{Lewis-etal-86}.  

\begin{remark}
While allowing variations of $\gamma_\pm$ adds generality to the Hamiltonian formulation of the IL$^{(0,1)}$QG, the practical consequences of such added generality are minimal.  This is illustrated in \textup{\cite{Ripa-JFM-93}} for the stability problem.    
\end{remark}

Let
\begin{align}
   \mathcal E[\bar\xi,\psi_\sigma,\psi_{\sigma^2}] &:= \frac{1}{2}\int_{\mathscr D}|\nabla\bar\psi|^2 +   R_S^{-2}\bar\psi^2\,d^2x\nonumber\\ &= - \frac{1}{2}\int_{\mathscr D} \bar\psi \left(\bar\xi - R_S^{-2}\left(\psi_\sigma-\tfrac{2}{3}\psi_{\sigma^2}\right) - \beta y\right)\,d^2x
\label{eq:E}
\end{align}
be the energy, where the equality follows from \eqref{eq:IL01-BCs} upon integration by parts. That $\mathcal E$ indeed is a functional of $(\bar\xi,\psi_\sigma, \psi_{\sigma^2})$ follows by noting that
\begin{equation}
   \bar\psi = (\nabla^2 - R_S^{-2})^{-1}\big(\bar\xi - R_S^{-2}(\psi_\sigma-\tfrac{2}{3}\psi_{\sigma^2}) - \beta y\big),   
\end{equation}
where $(\nabla^2 - R_S^{-2})^{-1}$ is interpreted in terms of the relevant Green’s function of the elliptic problem (\ref{eq:IL01}d).

Now, consider 
\begin{align}
   \{\mathcal F, \mathcal G\}
   &:= 
   \int_{\mathscr D}
   \begin{pmatrix}
      \frac{\delta \mathcal F}{\delta\bar\xi}\\
      \frac{\delta \mathcal F}{\delta\psi_\sigma}\\
      \frac{\delta \mathcal F}{\delta\psi_{\sigma^2}}
   \end{pmatrix}^\top
   \begin{pmatrix}
      -\{\bar\xi,\cdot\}_{xy} & -\{\psi_\sigma,\cdot\}_{xy} & -\{\psi_{\sigma^2},\cdot\}_{xy}\\
      -\{\psi_\sigma,\cdot\}_{xy} & 0 & 0\\
      -\{\psi_{\sigma^2},\cdot\}_{xy} & 0 & 0
   \end{pmatrix}
   \begin{pmatrix}
      \frac{\delta\mathcal G}{\delta\bar\xi}\\
      \frac{\delta\mathcal G}{\delta\psi_\sigma}\\
      \frac{\delta\mathcal G}{\delta\psi_{\sigma^2}}\\
   \end{pmatrix}
   d^2x\nonumber\\
   &=
   \int_{\mathscr D} \bar\xi\left\{\frac{\delta\mathcal F}{\delta\bar\xi},\frac{\delta\mathcal G}{\delta\bar\xi}\right\}_{xy}+ \sum_{n=1}^2 \psi_{\sigma^n} \left(\left\{\frac{\delta\mathcal F}{\delta\bar\xi},\frac{\delta\mathcal G}{\delta\psi_{\sigma^n}}\right\}_{xy} + \left\{\frac{\delta\mathcal F}{\delta\psi_{\sigma^n}},\frac{\delta\mathcal G}{\delta\bar\xi}\right\}_{xy}\right) \,d^2x
   \label{eq:LP}
\end{align}
for any $\mathcal F,\mathcal G[\bar\xi,\psi_\sigma,\psi_{\sigma^2}]$, where the second equality follows upon integrating parts and requiring
\begin{equation}
   \nabla^\perp\frac{\delta \mathcal F}{\delta\bar\xi}\cdot \hat{\mathbf n}\vert_{\partial \mathscr D_\pm} = 0,\quad 
   \nabla^\perp\frac{\delta \mathcal F}{\delta\psi_{\sigma}}\cdot \hat{\mathbf n}\vert_{\partial \mathscr D_\pm} = 0,\quad
   \nabla^\perp\frac{\delta \mathcal F}{\delta\psi_{\sigma^2}}\cdot \hat{\mathbf n}\vert_{\partial \mathscr D_\pm} = 0,
   \label{eq:adm}
\end{equation}
for all $\mathcal F[\bar\xi,\psi_\sigma,\psi_{\sigma^2}]$. It should be noted that periodicity in $x$ accounts for the vanishing of the corresponding flux integral across the (oriented) boundary piece $\smash{\partial\tilde{\mathscr D}} := \smash{\partial\mathscr D\setminus \partial\mathscr D_-\cup \partial\mathscr D_+}$. Explicitly, these integrals are of the form $\smash{\int_{\partial\tilde{\mathscr D}}}ab\nabla^\perp c\cdot\hat{\mathbf n}ds$, where $a,b,c(\mathbf x,t)$ are $x$-periodic. Noting that $\mathbf nds = -d\mathbf x^\perp$, one computes $\smash{\int_{\partial\tilde{\mathscr D}}}ab\nabla^\perp c\cdot\hat{\mathbf n}ds = \smash{\int_0^W} ab(-\partial_yc)\vert_{x=L}dy + \smash{\int_W^0} ab(-\partial_yc)\vert_{x=0}dy = \smash{\int_0^W} \big(ab\partial_y c\vert_{x=0} - ab\partial_y c\vert_{x=L}\big)\,dy = 0$. 

\remove[RA]{Functionals satisfying \mbox{\eqref{eq:adm}} are referred to as admissible functionals in \mbox{\textup{\cite{McIntyre-Shepherd-87}}}. If boundary conditions \mbox{\eqref{eq:adm}} are not imposed, then an appropriate redefinition of the functional derivative is needed for the second equality in \mbox{\eqref{eq:LP}} to hold, yet not without the imposition of some condition on the boundary terms; \mbox{cf.\ \textup{\cite{Lewis-etal-86}}}.}

For all $\mathcal F,\mathcal G,\mathcal H[\bar\xi,\psi_\sigma,\psi_{\sigma^2}]$, the bracket defined in \eqref{eq:LP} satisfies: $\{\mathcal F + \mathcal G, \mathcal H\} = \{\mathcal F, \mathcal H\} + \{\mathcal G, \mathcal H\}$ (bilinearity); $\{\mathcal F,\mathcal G\} = - \{\mathcal G,\mathcal F\}$ (antisymmetry); $\{\{\mathcal F,\mathcal G\},\mathcal H\} + \operatorname{\circlearrowleft} = 0$ (Jacobi identity); and $\{\mathcal F\mathcal G,\mathcal H\} = \mathcal F\{\mathcal G,\mathcal H\} + \{\mathcal F,\mathcal H\}\mathcal G$ (Leibniz rule). An explicit proof of the Jacobi identity is given in \cite{Beron-21-POFb}. These properties make $\{\,,\hspace{.01cm}\}$ a Poisson bracket and by its linear dependence on $(\bar\xi,\psi_\sigma,\psi_{\sigma^2})$ it is classified as of Lie--Poisson type.  Such noncanonical brackets inherit the properties above from analogous properties of $\{\,,\hspace{.01cm}\}_{xy}$, the canonical bracket in $\mathbb R^2$.

Since 
\begin{equation}
    \delta\mathcal E = - \int_{\mathscr D} \bar\psi\delta\bar\xi + R_S^{-2}\bar\psi \delta\psi_\sigma + \tfrac{2}{3}R_S^{-2}\bar\psi\delta\psi_{\sigma^2}\,d^2x,
    \label{eq:dE}
\end{equation}
where the boundary terms have cancelled out by Assumption \ref{ass:gamma}, it follows that
\begin{equation}
    \frac{\delta \mathcal E}{\delta\bar\xi} = -\bar\psi,\quad
    \frac{\delta \mathcal E}{\delta\psi_\sigma} = R_S^{-2}\bar\psi,
    \quad\frac{\delta \mathcal E}{\delta\psi_{\sigma^2}} = -\tfrac{2}{3}R_S^{-2}\bar\psi.
\end{equation}
Equations (\ref{eq:IL01}a)--(\ref{eq:IL01}b) then follow as
\begin{equation} 
    \partial_t\bar\xi = \{\bar\xi, \mathcal E\},\quad
    \partial_t\psi_\sigma = \{\psi_\sigma, \mathcal E\},\quad
    \partial_t\psi_{\sigma^2} = \{\psi_{\sigma^2}, \mathcal E\}.
\end{equation}
In other words, the IL$^{(0,1)}$QG model constitutes a Lie--Poisson Hamiltonian system in the variables $(\bar\xi,\psi_\sigma,\psi_{\sigma^2})$ with Hamiltonian given by \eqref{eq:E} and Lie--Poisson bracket defined by \eqref{eq:LP}. More generally, one has
\begin{equation}
    \dot{\mathcal F} = \{\mathcal F,\mathcal E\}
    \label{eq:dotF}
\end{equation}
for all $\mathcal F[\bar\xi,\psi_\sigma,\psi_{\sigma^2}].$

The IL$^0$QG is also Lie--Poisson Hamiltonian, with Hamiltonian given by \eqref{eq:E} with $R_S$ replaced by $R$ and Lie--Poisson bracket given by \eqref{eq:LP} with $\psi_{\sigma^2}$ ignored. The resulting bracket is found in \cite{Warneford-Dellar-13}, which is a particular version of the one derived by \cite{Ripa-RMF-96} for a system more general than the IL$^0$QG.  A (Lie--Poisson) deformation of the IL$^0$QG bracket, i.e., with the same kernel 
(cf.\ Section \ref{sec:integrals}), appears in \cite{Holm-Luesink-21}. Interestingly, low-$\beta$ magnetohydrodynamics \cite{Morrison-Hazeltine-84} and incompressible, nonhydrostatic, Boussinesq fluid dynamics on a vertical plane share this bracket \cite{Benjamin-84}. 

\change[RA]{Geometric mechanics views ideal fluid motion on a fluid container $\mathscr D$, such as the zonal channel domain of interest here, as a mapping of initial conditions at time $t=0$, regarded as labels taking values on $\mathscr D$, or configuration manifold, which is acted upon by the Lie group of diffeomorphisms \mbox{\textup{\cite{Marsden-etal-84}}}.  This way motion on $\mathscr D$ is lifted to the group and identified with a curve on the group, whose left action on $\mathscr D$ produces the fluid trajectories. The tangent lift of the velocity on $\mathscr D$ is given by the action of the group on the right. This lives on the tangent space to the group at the identity, arranged to happen at $t=0$. This represents a vector space, which together with a binary operation, namely, the Lie bracket, forms the Lie algebra of the group.  The relevant Lie algebra for IL$^{(0,1)}$QG dynamics is $\mathfrak{sdiff}(\mathscr D)$, that is, the Lie algebra of the Lie group $\mathrm{SDiff}(\mathscr D)$ of area preserving diffeomorphisms of the zonal channel domain $\mathscr D$.  The corresponding vector space can thus be identified with that of smooth time-dependent functions on $\mathscr D$, $C^\infty(\mathscr D)$, and the Lie bracket with the canonical Poisson bracket, $\{\,,\hspace{.01cm}\}_{xy}$, given in \mbox{\eqref{eq:canonical}} \mbox{\textup{\cite{Weinstein-83}}}. By identifying the dual (with respect to the $L^2$ inner product) of $C^\infty(\mathscr D)$ with itself, the Lie--Poisson bracket \mbox{\eqref{eq:LP}} represents a product for a realization of a Lie enveloping algebra on functionals on the dual of $\mathfrak{sdiff}(\mathscr D) \ltimes C^\infty(\mathscr D)^2$, that is, the Lie algebra of the Lie group obtained by extending $\mathrm{SDiff}(\mathscr D)$ by semidirect product to the vector space $\mathrm{SDiff}(\mathscr D) \times C^\infty(\mathscr D)^2$ with the (induced) representation of $\mathfrak{sdiff}(\mathscr D)$ on $C^\infty(\mathscr D)$ given by $\{\,,\hspace{.01cm}\}_{xy}$ \mbox{\textup{\cite{Marsden-Morrison-84}}}.}{For a geometric mechanics interpretation of the IL$^{(0,1)}$QG, where its Lie--Poisson bracket \mbox{\eqref{eq:LP}} is associated with the Lie algebra of the group of area-preserving diffeomorphisms extended by a semidirect product with that of the vector field formed by two copies of the space of smooth functions---this arising from broken symmetry---the interested reader is referred to \mbox{\cite{Beron-Luesink-25}}.}

\subsection{Topographic forcing}\label{sec:topo}

Let $\psi_0(\mathbf x)$ be an arbitrary scalar function on the zonal channel domain $\mathscr D$ of length-squared-over-time units, conveniently assumed to satisfy $\nabla\psi_0 \times \hat{\mathbf n}\vert_{\partial \mathscr D} = 0$. Let $L_0 = \const$ have units of length. The Lie--Poisson Hamiltonian structure of the IL$^{(0,1)}$QG is not spoiled if (\ref{eq:C-IL01}a) is modified as
\begin{equation}
    \partial_t\bar\xi + \big\{\bar\psi,\bar\xi - R_S^{-2} (\psi_\sigma-\tfrac{2}{3}\psi_{\sigma^2})\big\}_{xy} = L_0^{-2}\{\psi_0,\psi_\sigma-\tfrac{2}{3}\psi_{\sigma^2}\}_{xy}.
    \label{eq:topo}
\end{equation}
\change[RA]{which may be viewed as a topographic forcing to $\bar\xi$.}{The term on the right-hand side of \mbox{\eqref{eq:topo}} may be viewed as forcing resulting from imposing an irregular topography on the rigid lid separating the active layer from the atmosphere above. This extends an earlier \mbox{\cite{Holm-etal-21}} interpretation of a similar term in the case where the active layer rests on bottom topography while having a free top boundary with the atmosphere.} System \eqref{eq:C-IL01} with (\ref{eq:C-IL01}a) replaced with \eqref{eq:topo} is obtained from the Lie--Poisson bracket \eqref{eq:LP} with the Hamiltonian given by
\begin{equation}
    \mathcal E_0[\bar\xi,\psi_\sigma,\psi_{\sigma^2}] := \frac{1}{2}\int_{\mathscr D} |\nabla\bar\psi|^2 + R_S^{-2}\bar\psi^2 - 2L_0^{-2}\psi_0(\psi_\sigma-\tfrac{2}{3}\psi_{\sigma^2})\,d^2x. 
\end{equation}
The change of the Kelvin circulation $\mathcal K$ along a material loop $\partial \mathscr D_t$ is given by $\smash{\int_{\mathscr D_t}} [R_S^{-2}\bar\psi + L_0^{-2}\psi_0, \psi_\sigma-\tfrac{2}{3}\psi_{\sigma^2}]\,d^2x$, rather than \eqref{eq:dotK}. As in the unforced case, i.e., with $L = 0$, $\mathcal K$ is conserved when $\partial \mathscr D_t$ is levelwise isopycnic and along the walls $\partial \mathscr D_\pm$ of $\mathscr D$. 

\subsection{Conservation laws}\label{sec:integrals}

\subsubsection{Energy, zonal momentum, and Casimirs}\label{sec:energy-momentum-casimir}

The noncanonical Hamiltonian formalism enables the connection of conservation laws with symmetries through \emph{Noether's theorem} as follows:

\begin{theorem}[Noether for noncanonical Hamiltonians]\label{thm:noether}
    Let $P$ be a \emph{Poisson manifold}, i.e., an infinite-dimensional smooth manifold endowed with a Poisson bracket or structure $\{\,,\hspace{.01cm}\} : C^\infty(P) \times C^\infty(P) \to C^\infty(P)$.  The evolution of any functional $\mathcal F[\varphi] \in C^\infty(P) : P \to \mathbb R$ is controlled by $\smash{\dot{\mathcal F}} = \{\mathcal F, \mathcal E\}$ with the dynamics specified by $\mathcal E[\varphi]$, the Hamiltonian.  Consider the one-parameter family of variations induced by a functional $\mathcal G[\varphi]$ defined by $\delta_{\mathcal G} := - \varepsilon \{\mathcal G, \cdot\}$, where $\varepsilon > 0$ is small. The change induced by $\mathcal G$ on $\mathcal E$ is
    \begin{equation}
        \Delta_\mathcal G\mathcal E := \mathcal E[\varphi + \delta_\mathcal G\varphi] - \mathcal E[\varphi] \sim \varepsilon\{\mathcal E,\mathcal G\} = - \varepsilon \dot{\mathcal G},
    \end{equation}
    where $\sim$ means asymptotically as $\varepsilon\downarrow 0$. Moreover,
    \begin{equation}
        \frac{d}{dt}\Delta_\mathcal G \mathcal F - \Delta_\mathcal G\frac{d\mathcal F}{dt} \sim \varepsilon\left\{\mathcal F,\frac{d\mathcal G}{dt}\right\} 
        \label{eq:noether}
    \end{equation}
    for any $\mathcal F[\varphi]$. The following holds:
    \begin{enumerate}
        \item Symmetry of $\mathcal E$ under the transformation induced by $\mathcal G$ implies conservation of $\mathcal G$ and vice versa.
        
        \item If $\mathcal G$ is an integral of motion, then the transformation generated by $\mathcal G$ represents a symmetry in the general sense that the result of making a transformation and letting the time run is independent of the order in which these operations are performed.
    \end{enumerate}
\end{theorem}

A statement of Theorem \ref{thm:noether}(i) can be found in \cite{Shepherd-90}.  Theorem \ref{thm:noether}(ii) is due to \cite{Ripa-RMF-92a}. In both cases, if $\mathcal G$ is conserved, then it transforms solutions into solutions, i.e., $\partial_t\varphi = \{\varphi,\mathcal E\}$ transforms into $\partial_t\Delta_\mathcal G\varphi = \Delta_\mathcal G\partial_t\varphi  = \Delta_\mathcal G\{\varphi,\mathcal E\} = \{\Delta_\mathcal G\varphi,\mathcal E\}$ under the transformation induced by $\mathcal G$.

\remove[RA]{The Poisson manifold $P$ for the IL$^{(0,1)}$QG is given by the geometric dual of $\mathfrak{sdiff}(\mathscr D) \ltimes C^\infty(\mathscr D)^2$; \mbox{cf.\ Remark \ref{rem:geo}}.}

Note that $-\mathcal E$ is the generator of time shifts $t\to t - \varepsilon$ since $\delta_{-\mathcal E} \varphi = \varepsilon\{\mathcal E, \varphi\} = -\varepsilon\partial_t\varphi$.  By Theorem \ref{thm:noether}(i), conservation of $\mathcal E$ is related to symmetry of $\mathcal E$ under selfinduced transformations, as is the case of the IL$^{(0,1)}$QG's energy, given in \eqref{eq:E}. Let $\mathcal M[\varphi]$, identified with the \emph{zonal momentum}, be defined as the generator of the transformation $x\to x - \varepsilon$.  That is, $\delta_\mathcal M\varphi = - \varepsilon \{\mathcal M, \varphi\} := -\varepsilon\partial_x\varphi$, which is satisfied in the IL$^{(0,1)}$QG, i.e., with $\varphi = (\bar\xi, \psi_\sigma, \psi_{\sigma^2})$ and bracket \eqref{eq:LP}, by
\begin{equation}
    \mathcal M[\bar\xi, \psi_\sigma, \psi_{\sigma^2}] = \int_{\mathscr D} y\bar\xi\,d^2x.  
\end{equation}
Theorem \ref{thm:noether}(i) links constancy of $\mathcal M$ with symmetry of the IL$^{(0,1)}$QG's energy \eqref{eq:E} under zonal translations.

By Theorem \ref{thm:noether}(ii), conservation of $\mathcal E$ and $\mathcal M$ imply that the transformations they induced on any functional represents a symmetry in the general sense that transforming and letting the time run commute.  But the reciprocal of the latter is not strictly true.  Indeed, if the left-hand-side of \eqref{eq:noether} vanishes, then $\dot{\mathcal G} = 0$ is equal to a \emph{distinguished function}. Namely, a function of functionals $\mathcal C[\varphi]$, called \emph{Casimirs}, which are nontrivial solutions of 
\begin{equation}
    \{\mathcal F, \mathcal C\} = 0\,\forall\mathcal F.   
\end{equation}
The Casimirs form the kernel of $\{\,,\hspace{.01cm}\}$ and, since $\mathcal F$ includes $\mathcal E$, they represent integrals of motion. However, they do not generate any variation, i.e., $\delta_\mathcal C\varphi = -\varepsilon\{\mathcal C,\varphi\} = 0$, and thus are not related to explicit symmetries, i.e., symmetries that are visible in the  Eulerian variables that form the phase space in which the dynamics are being viewed. The Casimirs of the IL$^{(0,1)}$QG's bracket \eqref{eq:LP} are given by
\begin{equation}
    \mathcal C_{a,F}[\bar\xi, \psi_\sigma, \psi_{\sigma^2}] := \int_{\mathscr D}  a\bar\xi + F(\psi_\sigma, \psi_{\sigma^2})\,d^2x
    \label{eq:C-IL01} 
\end{equation}
for any constant $a$ and arbitrary function $F$.  In \cite{Beron-21-POFb} a proof for the IL$^{(0,1)}$PE's Casimirs was given (cf.\ Remark \ref{rem:C-IL01PE}, below). Here we present the proof for the QG case. A genuine Casimir $\mathcal C[\bar\xi,\psi_\sigma,\psi_{\sigma^2}]$ must satisfy:
\begin{equation}
    \begin{pmatrix}
        \{\bar\xi,\cdot\,\}_{xy} & \{\psi_\sigma,\cdot\,\}_{xy} & \{\psi_{\sigma^2},\cdot\,\}_{xy}\\
        \{\psi_\sigma,\cdot\,\}_{xy} & 0 & 0\\
       \{\psi_{\sigma^2},\cdot\,\}_{xy} & 0 & 0
    \end{pmatrix}
    \begin{pmatrix}
        \frac{\delta \mathcal C}{\delta\bar\xi}\\
        \frac{\delta \mathcal C}{\delta\psi_\sigma}\\
        \frac{\delta \mathcal C}{\delta\psi_{\sigma^2}}
    \end{pmatrix}
    =
    \begin{pmatrix}
        0\\ 0\\ 0
    \end{pmatrix}
    .
\end{equation}
From the last two rows, it is clear that $\smash{\frac{\delta\mathcal C}{\delta\bar\xi}}$ must be a constant. This makes the first term in the first row to vanish.  The last two terms vanish by adding the integral of an arbitrary function of $(\psi_\sigma,\psi_{\sigma^2})$, showing that \eqref{eq:C-IL01} indeed is a valid Casimir.  In fact, there exists no Casimir more general than \eqref{eq:C-IL01}. 

\begin{remark}
    The admissibility conditions \eqref{eq:adm} are satisfied by the Casimir $\mathcal C_{a,F}$, given in \eqref{eq:C-IL01}, provided that $\psi_\sigma$ and $\psi_{\sigma^2}$ are constant along $\partial\mathscr D_\pm$, that is, provided that the coasts of the zonal channel domain $\mathscr D$ are isopycnic.  This same condition must hold in the IL$\smash{^{(0,1)}}$PE for its Casimirs, given by $\mathcal C_{a,F}[\bar{\mathbf u}, h, \bar\vartheta, \vartheta_\sigma] := \smash{\int_{\mathscr D}}  a\bar q h+ F(\bar\vartheta, \vartheta_\sigma)\,d^2x$ for $a=\const$ and every $F$ where $\bar q$ is the potential vorticity, to commute with any functional in the corresponding bracket.  This Casimir family reduces to $\mathcal C_{a,F}$ in the QG limit.
    \label{rem:C-IL01PE}
\end{remark}

The conservation laws above can all be verified directly.  Indeed, after integrating by parts with boundary conditions \eqref{eq:IL01-BCs} in mind, one computes
\begin{equation}
    \dot{\mathcal E} = \bar\psi\smash{\vert_{\partial \mathscr D_-}} \dot\gamma_- + \bar\psi\smash{\vert_{\partial \mathscr D_+}}\dot\gamma_+,
\end{equation}
\emph{which vanishes when the velocity circulations along the walls of the zonal channel domain are constant}, justifying condition \eqref{eq:dotgamma}.  As for the zonal momentum, one has that
\begin{equation}
    \dot{\mathcal M} = \int_{\mathscr D} (\nabla^2\bar\psi - R_S^{-2}\bar\psi + \beta y)\partial_x\bar\psi\,d^2x = 0,
\end{equation}
where the first equality follows upon cancellation of boundary terms by virtue of \eqref{eq:IL01-BCs} and the second one by using (\ref{eq:IL01-BCs}b), in particular. In turn, verifying the constancy of $\mathcal C_{a,F}$ directly is a straightforward application of boundary conditions \eqref{eq:IL01-BCs}.

Finally, the IL$^0$QG has the same motion integrals as above with $R_S$ replaced by $R$ in the \eqref{eq:E} and, in particular, with \eqref{eq:C-IL01} replaced by
\begin{equation}
    \mathcal C_{F,G}[\bar\xi, \psi_\sigma, \psi_{\sigma^2}] := \int_{\mathscr D}  \bar\xi F(\psi_\sigma) + G(\psi_\sigma)\,d^2x
    \label{eq:C-IL0} 
\end{equation}
for arbitrary $F,G$, which commutes with any functional with the appropriate Lie--Poisson bracket. This Casimir is well known \cite{Morrison-Hazeltine-84, Benjamin-84}. To verify conservation of energy and zonal momentum directly one simply proceeds as above.  To check conservation of \eqref{eq:C-IL0} directly one proceeds similarly as we do to prove constancy of \eqref{eq:I}, below.

\begin{remark}
    In a similar manner as for the Casimir family \eqref{eq:C-IL01} to be genuine, genunity of the Casimir family \eqref{eq:C-IL0} imposes a condition on the boundaries of the flow's zonal channel domain: its coasts must be isopycnic.
\end{remark}

\subsubsection{Weak Casimirs}

\begin{lemma}
   The dynamics of the IL$^{(0,1)}$QG preserve the following infinite family of functionals    
   \begin{equation}
      \mathcal I_F[\bar\xi, \psi_\sigma, \psi_{\sigma^2}] := \int_{\mathscr D}  \bar\xi F(\psi_\sigma - \tfrac{2}{3}\psi_{\sigma^2})\,d^2x
      \label{eq:I}
   \end{equation}
   where $F$ is arbitrary.
   \label{lem:extra}
\end{lemma}
\begin{proof}
   Since $\nabla\cdot\nabla^\perp\bar\psi = 0$, from (\ref{eq:IL01}b)--(\ref{eq:IL01}c) one has
   \begin{equation}
      \partial_t G(\psi_\sigma, \psi_{\sigma^2}) + \nabla\cdot G(\psi_\sigma, \psi_{\sigma^2}) \nabla^\perp\bar\psi = 0
      \label{eq:k}
   \end{equation}
   for every $G$. Integrating over $\mathscr D$,
   \begin{equation}
      \int_{\mathscr D} G(\psi_\sigma, \psi_{\sigma^2})\,d^2x = \const
      \label{eq:G}
   \end{equation}
   because the flow normal to the channel walls vanishes and the there is periodicity in the along-channel direction.  Now define $\chi := \psi_\sigma - \tfrac{2}{3}\psi_{\sigma^2}$. Multiplying (\ref{eq:IL01}a) by $F(\phi)$ for any $F$ and using \eqref{eq:k} with $G(\psi_\sigma, \psi_{\sigma^2})$ replaced by $F(\chi)$,
   \begin{equation}
      \partial_t\left(\bar\xi F(\chi)\right) + \nabla\cdot \bar\xi F(\chi) \nabla^\perp\bar\psi = R_S^{-2}F(\chi)\nabla\cdot \phi\nabla^\perp\bar\psi.
      \label{eq:F}
   \end{equation}
   Multiplying \eqref{eq:G} by $R_S^{-2}$ with $G(\psi_\sigma, \psi_{\sigma^2})$ replaced by $\int F(\chi)\,d\chi$, and adding the result to \eqref{eq:k}, upon integrating over $\mathscr D$,
   \begin{equation}
      \int_{\mathscr D} \bar\xi F(\chi) + R_S^{-2}\int F(\chi) \,d\chi\,d^2x = \const.
      \label{eq:I2}
   \end{equation}
   More generally, with \eqref{eq:G} in mind, 
   \begin{equation}
      \int_{\mathscr D} \bar\xi F(\chi) + R_S^{-2}\int F(\chi) \,d\chi + G(\psi_\sigma,\psi_{\sigma^2})\,d^2x = \const,
      \label{eq:I3}
   \end{equation}
   from which conservation of \eqref{eq:I} follows by choosing $G(\psi_\sigma,\psi_{\sigma^2}) = - R_S^{-2}\int F(\chi) \,d\chi$ for arbitrary $F$.
\end{proof}

Three salient observations are warranted:
\begin{enumerate}
    \item The integrals of motion \eqref{eq:I} do \emph{not} constitute Casimirs. Indeed,
    \begin{equation}
       \begin{pmatrix}
          \{\bar\xi,\cdot\,\}_{xy} & \{\psi_\sigma,\cdot\,\}_{xy} & \{\psi_{\sigma^2},\cdot\,\}_{xy}\\
          \{\psi_\sigma,\cdot\,\}_{xy} & 0 & 0\\
          \{\psi_{\sigma^2},\cdot\,\}_{xy} & 0 & 0
        \end{pmatrix}
      \begin{pmatrix}
         \frac{\delta \mathcal I_F}{\delta\bar\xi} = F\\
         \frac{\delta \mathcal I_F}{\delta\psi_\sigma} = \bar\xi F'\\
         \frac{\delta \mathcal I_F}{\delta\psi_{\sigma^2}} = -\tfrac{2}{3}\bar\xi F'
      \end{pmatrix}
      =
      \begin{pmatrix}
         0\\
         -\tfrac{2}{3}\{\psi_\sigma,\psi_{\sigma^2}\}_{xy}F'\\
         \{\psi_{\sigma^2},\psi_\sigma\}_{xy}F'
      \end{pmatrix},
    \end{equation}
    where the last two rows in general do not vanish. That is, $\mathcal I_F$ does not commute with every functional in \eqref{eq:LP}.

    \item While conservation laws \eqref{eq:I2} and \eqref{eq:I3} are more general than $\mathcal I_F$, it is the first term in each case, given by $\mathcal I_F$, that does not represent a Casimir of the Lie--Poisson bracket \eqref{eq:LP}.

    \item The family of conservation laws $\mathcal I_F$ constitute Casimirs, but for the IL$^0$QG bracket, namely,
    \begin{equation}
       \{\mathcal F, \mathcal G\}_{\bar\xi\chi} := \int_{\mathscr D} \bar\xi\left\{\frac{\delta\mathcal F}{\delta\bar\xi},\frac{\delta\mathcal G}{\delta\bar\xi}\right\}_{xy} + \chi \left(\left\{\frac{\delta\mathcal F}{\delta\bar\xi},\frac{\delta\mathcal G}{\delta\chi}\right\}_{xy} + \left\{\frac{\delta\mathcal F}{\delta\chi},\frac{\delta\mathcal G}{\delta\bar\xi}\right\}_{xy}\right) \,d^2x,
       \label{eq:LPphi}
    \end{equation}
    for all $\mathcal F,\mathcal G[\bar\xi,\chi]$. The IL$^{(0,1)}$QG dynamics implies the dynamics in the variables $(\bar\xi,\chi)$, i.e., that produced by the Lie--Poisson bracket given by \eqref{eq:LPphi} with Hamiltonian \eqref{eq:E}, yet not \emph{vice versa}. 
\end{enumerate}

The last observation above enables the interpretation of the conservation laws $\mathcal I_F$ as weak Casimirs as they form the kernel of the Lie–Poisson bracket for the potential vorticity evolution independent of the details of the buoyancy field as this field is advected under the flow.  \add[RA]{The connection of conservation of $\mathcal I_F$ with particle relabeling symmetry via Noether's theorem is discussed in \mbox{\cite{Beron-Luesink-25}}.}

\begin{remark}
    The IL$^{(0,1)}$PE conserves $\smash{\int_{\mathscr D}} \bar qhF(\bar\vartheta - \smash{\tfrac{1}{3}}\vartheta_\sigma)\,d^2x$ for every $F$ (where $\bar q$ is the potential vorticity).  \change[]{These reduce to \mbox{\eqref{eq:I}}.}{These conservation laws reduce to \mbox{\eqref{eq:I}} in the low-frequency limit.}  The difference between the $\smash{\tfrac{1}{3}}$ and $\smash{\tfrac{2}{3}}$ factors is clarified upon noting that, with an $\smash{O(\mathrm{Ro}^2)}$ error, $\bar\vartheta - g_\mathrm{r} - \smash{\tfrac{1}{3}}(\vartheta_\sigma - \tfrac{1}{2}N^2_\mathrm{r}H_\mathrm{r}) = \smash{\frac{2g_\mathrm{r}}{f_0R^2}}(\psi_\sigma - \smash{\tfrac{2}{3}}\psi_{\sigma^2})$. The conservation laws above are not Casimirs either.  However, they represent weak Casimirs in the sense previously discussed.  \remove[]{We speculate that their existence should be attributed to invariance of the Eulerian variables under relabeling of fluid particles. The Noether theorem that should lead to this conclusion differs from the one discussed above as it should arise from the Euler--Poincar\'e variational principle with variations not restricted to vanish at the endpoints of the action integral, as considered by \mbox{\textup{\cite{Cotter-Holm-12}}} for other systems.}
\end{remark}

\section{Linear waves and free energy}\label{sec:waves}

Let $\varphi' := (\bar\xi',\psi'_\sigma,\psi'_{\sigma^2})$ be an infinitesimal-amplitude perturbation on 
\begin{equation}
    \Phi_\mathrm{r} := (- R_S^{-2}(a_1 - a_2 + \tfrac{2}{3}a_3) + \beta y, a_2, a_3),
    \label{eq:ref}
\end{equation}
where $a_1,a_2,a_3$ are arbitrary constants. The latter represents a three-parameter family of \emph{reference states}, i.e., stationary or equilibrium solutions of the IL$^{(0,1)}$QG system \eqref{eq:IL01}, subjected to \eqref{eq:IL01-BCs} and \eqref{eq:dotgamma}, with \emph{no} currents. The energy of such states, $\mathcal E[\Phi_\mathrm{r}] = \frac{1}{2}R_S^{-2}a_1^2WL$. This identically vanishes for the particular class of reference states with $a_1 = 0$.  Such a two-parameter family of reference states has the lowest possible energy and can thus be identified with the \emph{vacuum state} \cite[e.g.,][]{Morrison-Eliezer-86} of the IL$^{(0,1)}$QG.  The components of the perturbation field $\varphi'$ evolve according to
\begin{equation}
    \partial_t\bar\xi' + \beta\partial_x\bar\psi' = 0,\quad
    \partial_t\psi'_{\sigma} = 0,\quad
    \partial_t\psi'_{\sigma^2} = 0,
\end{equation}
controlling the IL$^{(0,1)}$QG linearized dynamics around \eqref{eq:ref}.  Assuming an $(x,y,t)$-dependence of the form of a normal mode, viz., $\exp\mathrm{i}(kx - \omega t)\sin ly$, satisfying the boundary conditions \eqref{eq:IL01-BCs}, it follows that $\psi'_\sigma = 0 = \psi'_{\sigma^2}$ while $\bar\xi' = \nabla^2\bar\psi' - R_S^{-2}\bar\psi' \neq 0$ provided that
\begin{equation}
    \omega = - \frac{k\beta}{k^2 + l^2 + R_S^{-2}}.
\end{equation}
That is, the linear waves supported by the IL$^{(0,1)}$QG are the well-known (planetary) Rossby waves \cite[e.g.,][]{Pedlosky-82}, which do not disturb the reference buoyancy field $g_\mathrm{r} + (1 + \smash{\frac{2z}{H_\mathrm{r}}})\frac{1}{2}N^2_\mathrm{r}H_\mathrm{r}$. They correspond to the Rossby waves of the HLQG in the weak-stratification limit $S\downarrow 0$.  For finite $S$, the (westward) phase speed is in general slower.  For instance, in the infinitely wide zonal channel case, the fastest phase speed, realized in the short-wave limit $(k,l)\downarrow (0,0)$, is $\beta (1 - \tfrac{1}{3}S)R^2$, which can be up to one-third slower than that in the HLQG, given by $\beta R^2$. In addition to the Rossby waves, the IL$\smash{^{(0,1)}}$QG supports an $\omega = 0$, i.e., neutral, mode with $\bar\xi' = R^{-2}_1(\psi'_\sigma - \frac{2}{3}\psi'_{\sigma^2})$ with $\psi'_\sigma, \psi'_{\sigma^2} \neq 0$.  This neutral mode is characterized by potential vorticity changes exclusively due to buoyancy changes which do not alter the fluid velocity ($\bar\psi' = 0)$.

One test of the well-posedness of the the IL$^{(0,1)}$QG is given by the computation of the \emph{free energy} \cite{Ripa-JFM-95}.  Denoted by $\mathcal E_\text{free}$, this is defined as an integral of motion quadratic to the lowest-order in the deviation from a reference state.  When $\mathcal E_\text{free}$ is positive definite, there cannot be a spurious growth of the amplitude of linear waves riding on the reference state as they must preserve $\mathcal E_\text{free}$.  A class of reference states \eqref{eq:ref} with a free energy with the desired property is defined by $\smash{\frac{a_1}{a_2}} < 0$ and $\smash{\frac{a_1}{a_3}} > 0$.  Indeed, for such a class of reference states,
\begin{equation}
    \mathcal E_\text{free}[\bar\xi',\psi'_\sigma,\psi'_{\sigma^2}] = \mathcal E[\bar\psi'] + R_S^{-2}\int_{\mathscr D} \frac{2}{3}\frac{a_1}{a_3}\psi^{'2}_{\sigma^2} - \frac{a_1}{a_2}\psi^{'2}_\sigma\,d^2x,
    \label{eq:free}
\end{equation}
which is manifestly positive definite under the stated conditions.  This follows by adding to the energy \eqref{eq:E} a Casimir \eqref{eq:C-IL01} with $a = a_1$ and $F(\psi_\sigma,\psi_{\sigma^2}) = R_S^{-2}(\smash{\frac{a_1}{3a_3}}\psi_{\sigma^2}^2 - \smash{\frac{a_1}{2a_2}}\psi_\sigma^2)$.  An important observation is that \eqref{eq:free} is not restricted to infinitesimally small normal-mode perturbations.  In fact, \eqref{eq:free} is an exact integral of motion for the fully nonlinear dynamics about \eqref{eq:ref}.  Thus positive definiteness of \eqref{eq:free} prevents the spurious growth of perturbations to \eqref{eq:ref} irrespective of their initial amplitude and structure. 

It is fair to wonder about the dynamical significance of the constants $a_1,a_2,a_3$ that define the family of reference states \eqref{eq:ref} and the relationship among them for $\mathcal E_\text{free}$ to be positive definite.  Consider the situation in which these constants are all set to zero, in which case \eqref{eq:ref} reduces to $\Phi_\mathrm{r} = (\beta y, 0, 0)$, a member of the vacuum state family.  The linear waves riding on this reference state are the same as those discussed above.  However, in this case, the free energy can be shown to be equal to $\mathcal E[\bar\psi']$, which is positive semidefinite. Indeed, by the invertibility relationship \eqref{eq:inv}, it follows that there can be perturbations $(\bar\xi',\psi'_\sigma,\psi'_{\sigma^2})$ that do not change the free energy. Spontaneous growth of the amplitude of such perturbations cannot be constrained by free-energy conservation.  An example is the neutral mode discussed above, for which the free energy identically vanishes. The same considerations apply to the complete family of vacuum states (defined by $a_1 = 0$).  An important observation is that there exist Lyapunov stable states that will arrest the aforementioned spontaneous growth, should it happened, as discussed in Section \ref{sec:bounds}, below.

\begin{remark}
    A result similar to that for perturbations on the vacuum state family in the IL$^{(0,1)}$QG holds for perturbations on the vacuum state in the IL$^{(0,1)}$PE parent model.  In that model, there is a unique reference state, which is the vacuum state, and perturbations on it have a positive-semidefinite free energy associated with \textup{\cite[cf.][]{Ripa-JFM-95}}.
\end{remark}

\section{Stability of baroclinic zonal flow}\label{sec:stability}

Consider $\varphi := (\bar\xi,\psi_\sigma,\psi_{\sigma^2})$, evolving according to \eqref{eq:IL01} subjected to \eqref{eq:IL01-BCs} and \eqref{eq:dotgamma}.  We call $\Phi := (\bar\Xi,\Psi_\sigma,\Psi_{\sigma^2})$ a \emph{basic state} if it is an equilibrium of \eqref{eq:IL01}, subjected to \eqref{eq:IL01-BCs} and \eqref{eq:dotgamma}.  A particular case of basic state is the reference state discussed in Section \ref{sec:waves}, for which the fluid is motionless.  We are interested in the stability of a basic state in the following three senses, sorted in increasing level of strength \cite{Holm-etal-85,Ripa-RMF-00}: 
\begin{enumerate}
    \item The first sense is \emph{spectral stability}, in which case a perturbation $\delta\varphi(\mathbf x,t)$ on $\Phi(\mathbf x)$ is assumed to be a normal mode with an infinitesimally small amplitude, viz., $\delta\varphi = \varepsilon\mathrm{Re}\{\varphi' \exp\mathrm{i}k(x-ct)\} \sin ly + O(\varepsilon^2)$ where $\varepsilon \downarrow 0$, satisfying \eqref{eq:IL01-BCs} and \eqref{eq:dotgamma}.  (Note the slight change of interpretation of the prime notation with respect to that used in Section \ref{sec:waves}.) Clearly, the amplitude of a normal mode cannot remain infinitesimally small unless $\mathrm{Im}(c) = 0$. In such a case, the basic state is said to be spectrally stable.
    
    \item The second sense is \emph{formal stability}. In this case $\delta\varphi$, while assumed small, can have an arbitrary structure. If there exists an integral of motion $\mathcal H[\varphi]$ such that $\delta \mathcal H[\delta\varphi;\Phi] = 0$ and $\delta^2\mathcal H[\delta\varphi;\Phi] > 0$, then the growth of $\delta\varphi$ will be constrained. The constraint is imposed by the quadratic nature of $\delta^2\mathcal H$, which is preserved by the linearized dynamics about $\Phi$.  In these circumstances one says that $\Phi$ is formally stable.  Note that the free energy $\mathcal E_\text{free}$ discussed in Section \ref{sec:waves} is a special case of $\delta^2\mathcal H$ wherein $\Phi$ in motionless.
    
    \item The third sense is \emph{Lyapunov stability}. This relies on the possibility of proving that the total variation $\Delta \mathcal H[\delta\varphi;\Phi]$ is convex for finite-size $\delta\varphi$ of arbitrary structure.  This means showing that $c_1\|\delta\varphi\|^2 \le \Delta \mathcal H \le c_2\|\delta\varphi\|^2$ for $c_1,c_2 = \const$ such that $0 < c_1 \le c_2 < \infty$, where $\|\,\|$ is typically chosen to be an $L^2$ norm.  Assuming that the latter holds true, since $\Delta \mathcal H$ is preserved under the fully nonlinear dynamics, the second inequality can be evaluated at time $t=0$ to get $\|\delta\varphi\|_{t>0} \le \smash{\sqrt{\frac{c_2}{c_1}}} \|\delta\varphi\|_{t=0}$. This implies Lyapunov stability for $\Phi$. More precisely, for every $\varrho > 0$ there exists $\varepsilon > 0$, e.g., $\varepsilon = \smash{\sqrt{\frac{c_1}{c_2}}}\varrho$, such that $\|\delta\varphi\|_{t=0} < \varepsilon$ implies $\|\delta\varphi\|_{t>0} < \varrho$.
\end{enumerate}

Let
\begin{subequations}
\begin{equation}
   \alpha := \frac{\bar U}{U_\sigma},\quad
   \mu := \frac{U_{\sigma^2}}{U_\sigma},\quad
   b := \frac{\beta}{U_\sigma R^2},
\end{equation}
where $\bar U$,  $U_\sigma$, and $U_{\sigma^2}$ are constants, with the latter two required to satisfy
\begin{equation}
    2U_\sigma - U_{\sigma^2} \ge \frac{(S-1)f_0R^2}{2W},\quad
    U_{\sigma^2} < \frac{Sf_0R^2}{4W},
\end{equation}
by static stability; cf.\ \eqref{eq:vertstab}. Our main interest is the four-parameter family of basic states defined by
\begin{equation}
    \bar\Psi = -\bar Uy,\quad
    \Psi_\sigma = - U_\sigma y,\quad
    \Psi_{\sigma^2} = - U_{\sigma^2} y,
\end{equation}
implying
\begin{equation}
    \bar\Xi = - R_S^{-2}\bar\Psi + R_S^{-2} \big(\Psi_\sigma - \tfrac{2}{3}\Psi_{\sigma^2}\big) + \beta y = R_S^{-2}\nu\Psi_\sigma = R_S^{-2}\mu^{-1}\nu\Psi_{\sigma^2},
\end{equation}
where
\begin{equation}
    \nu(S) := 1 + \alpha + s(S)b - \tfrac{2}{3}\mu,\quad
    s(S) := R_S^2R^{-2} = 1 - \tfrac{1}{3}S,   
\end{equation}
\label{eq:BS}%
\end{subequations}
which we introduce to simplify algebraic expressions below. By the thermal-wind balance, \eqref{eq:BS} implicitly corresponds, for parameters fixed, to a meridionally uniform zonal current with quadratic vertical shear, i.e., a baroclinic flow with \emph{vertical curvature} on the $\beta$-plane.  Nondimensional parameters $\alpha$ and $\mu$ thus measure the vertical linear shear and curvature of the jet, respectively. But, by virtue of \eqref{eq:h}--\eqref{eq:vartheta_sigma}, one also has that
\begin{equation}
    \alpha = 1 + \frac{2sg_\mathrm{r}}{H_\mathrm{r}}\frac{H'(y)}{\bar\Theta'(y)} - \frac{H_\mathrm{r}}{6}\frac{(N^2)'(y)}{\bar\Theta'(y)},\quad
    \mu = \frac{H_\mathrm{r}}{4}\frac{(N^2)'(y)}{\bar\Theta'(y)},
    \label{eq:alpha-mu}
\end{equation}
where $H$ is the layer thickness in the basic state, $\bar\Theta$ is the basic vertically averaged buoyancy, and $N^2 = 2\smash{\frac{\Theta_{\sigma}}{H_\mathrm{r}}}$ is the basic state buoyancy frequency squared. This enables further physical insight into the (spectral) stability problem. Finally, nondimensional parameter $b$ measures the strength of the $\beta$ effect, which can be identified with a Charney number.  

\subsection{Spectral stability}\label{sec:spectral}

Let
\begin{equation}
   \lambda := \frac{c - \bar U}{U_\sigma},\quad 
   \kappa := R\sqrt{k^2+l^2},
\end{equation}
representing normalized Doppler shifted phase speed and wavenumber $\mathbf k = (k,l)$ (with $k$ pointing eastward and $l$ northward) magnitude, respectively.  The following eigenvalue problem follows upon proposing a normal-mode perturbation on \eqref{eq:BS}:
\begin{equation}
    \begin{pmatrix}
        \left(s\kappa^2+1\right)\lambda + \alpha + sb & -(\lambda + \alpha) & \frac{2}{3}(\lambda + \alpha)\\
        -1 & -\lambda & 0\\
        -\mu & 0 & -\lambda
    \end{pmatrix}
    \begin{pmatrix}
        \bar\psi'\\
        \psi'_\sigma\\
        \psi'_{\sigma^2}
    \end{pmatrix} 
    =
    0.
\end{equation}
Nontrivial solutions exist provided that the determinant of the matrix of the eigenproblem vanishes.  This leads to the dispersion relation
\begin{equation}
    \lambda = \frac{-\nu \pm \sqrt{\nu^2 - 4\alpha\left(1 - \tfrac{2}{3}\mu\right)\left(s\kappa^2 + 1\right)}}{2 \left(s\kappa^2 + 1\right)}.
    \label{eq:disprel}
\end{equation}
Spectral stability is realized when the perturbation phase speed $c$, or equivalently $\lambda$, is real. A sufficient condition for this is
\begin{equation}
    \alpha\left(1 - \tfrac{2}{3}\mu\right) < 0.
    \label{eq:stab-spec}
\end{equation}
The shaded quadrants in Figure \ref{fig:alphamu} are subsets of the region in the $(\alpha,\mu)$-space defined by \eqref{eq:stab-spec}, namely,
\begin{equation}
    \{\alpha<0\} \cap \{\mu<\tfrac{3}{2}\} \bigcup \{\alpha>0\} \cap \{\mu>\tfrac{3}{2}\},
    \label{eq:stab-spec-set}
\end{equation}
where there is spectral stability of basic state \eqref{eq:BS} for every wavenumber. \emph{This holds independent of the values taken by the Charney number ($b$) and stratification parameter ($S$).} By virtue of \eqref{eq:alpha-mu}, one concludes that the basic states in the upper-right (resp., lower-left) clear quadrant have $H'(y)$ and $\bar\Theta'(y)$ with like (resp., opposing) signs and $\smash{\frac{\bar\Theta'(y)}{(N^2)'(y)}}$ smaller (resp., larger) than $\frac{1}{6}H_\mathrm{r}$. The curves labeled by $\kappa$ bound the (sub)regions of the $(\alpha,\mu)$-plane where there is spectral instability, i.e.,  \eqref{eq:stab-spec} is violated, for selected $b$ and $S$ values.  Note the asymmetry, with less basic states being spectrally unstable toward smaller wavenumbers in the region $\{\alpha < 0\} \cap \{\mu > \frac{3}{2}\}$ than in the region $\{\alpha > 0\} \cap \{\mu < \frac{3}{2}\}$. This behavior reverses with positive values of $b$.

\begin{figure}
    \centering 
    \includegraphics[width=.5\textwidth]{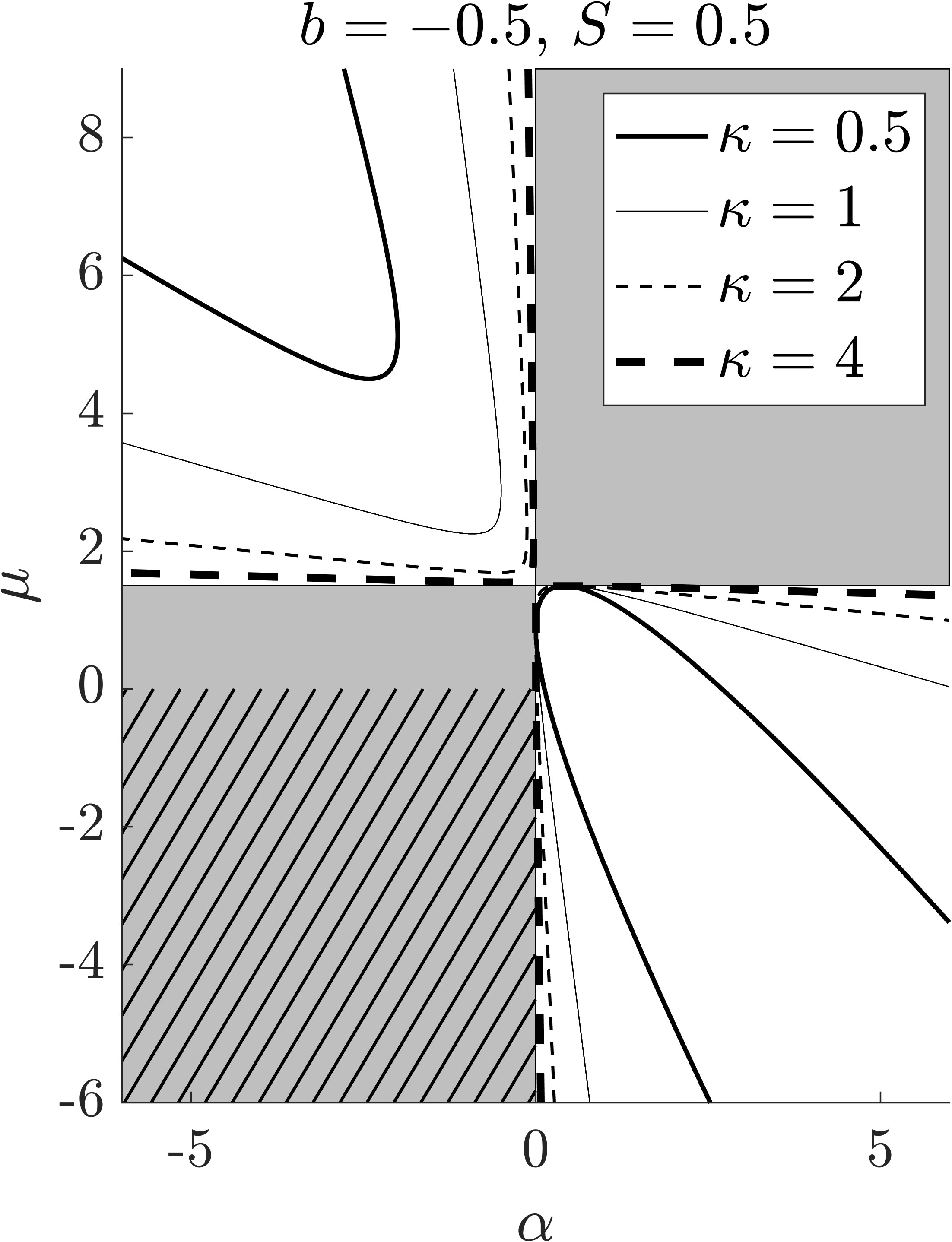}
    \caption{Stability of the basic state family \eqref{eq:BS}, representing baroclinic zonal jets, in the $\alpha := \bar U/U_\sigma$ vs $\mu := U_{\sigma^2}/U_\sigma$ space, respectively measuring the vertical linear shear and curvature of a jet in the family.  In the shaded regions, corresponding to the set $\mathbb S$ defined by $\alpha (1 - \frac{2}{3}\mu) \le 0$, the phase speed of an infinitesimally-small normal-mode perturbation is real for every wavenumber. Within the hatched region, the subset $\mathbb L:= \{\mu < 0\} \cap \{\alpha < 0\}$ of $\mathbb S$, there is stability for finite-size perturbations of arbitrary structure.  Moreover, in $\mathbb L$ the distance, in an $L^2$ sense, of a perturbation to the basic state is at all times bounded by a multiple of its initial distance. This means that in $\mathbb L$ there is Lyapunov stability.  Spectral stability is possible inside the white regions, the complement of $\mathbb S$, $\mathbb S^c$. The curves labeled by the normalized wavenumber magnitude $\kappa := R\smash{\sqrt{k^2 + l^2}}$ bound the $(\alpha,\mu)$-subregions of $\mathbb S^c$ where there is spectral instability for the normalized $\beta$ and stratification parameters, $b := \beta/U_\sigma R^2$ and $S := N^2_\mathrm{r} H_\mathrm{r}/2g_\mathrm{r}$, respectively, as indicated.}
    \label{fig:alphamu}
\end{figure}

More specifically, in the limit $\kappa \downarrow 0$, \eqref{eq:disprel-ass} equals
\begin{equation}
    \lambda = -\tfrac{1}{2}\nu \pm \tfrac{1}{2}\sqrt{\nu^2 - 4\alpha\left(1 - \tfrac{2}{3}\mu\right)},
\end{equation}
which is complex where $\alpha(1 - \frac{2}{3}\mu) > \frac{1}{4}\nu^2 = \frac{1}{4}\big(1 + \alpha + sb - \frac{2}{3}\mu\big)^2$ (and real otherwise, making explicit that \eqref{eq:stab-spec} is sufficient, yet not necessary, for spectral stability). At criticality, one finds that $\mu = \smash{\frac{3}{2}}(1 - \alpha + sb) \pm \smash{\frac{3}{2} \sqrt{(1 - \alpha + sb) - (1 + \alpha + sb)^2 + 4\alpha}}$. This represents a parabola with tangencies on the curve $C_+ := \smash{\{\alpha = 0, \mu > \frac{3}{2}\} \cup \{\alpha > 0, \mu = \frac{3}{2}\}}$, when $b < 0$, and the curve $C_- := \{\alpha = 0, \mu < \frac{3}{2}\} \cup \{\alpha < 0, \mu = \frac{3}{2}\}$, when $b > 0$, making more explicit the asymmetry noted above.  For completeness, we note that in the limit $\kappa \uparrow \infty$ there is spectral stability everywhere in the $(\alpha,\mu)$-space where \eqref{eq:stab-spec} is violated, in such a case for any $b$ and $S$. Indeed,  
\begin{equation}
    \lambda \sim \pm\mathrm{i}\sqrt{\frac{\alpha(1 - \tfrac{2}{3}\mu)}{1 - \tfrac{1}{3}S}}\kappa^{-1}
    \label{eq:disprel-ass}
\end{equation}
asymptotically as $\kappa\uparrow\infty$. The curves $\kappa \uparrow \infty$ in the $(\alpha,\mu)$-plane are given by $C_\pm$.  Their union $\{\alpha = 0, \mu = \frac{3}{2}\}$ corresponds to the set of basic states lying at the boundary of spectral stability. Thus, unlike all other curves in the spectrally unstable regions of the $(\alpha,\mu)$-plane, the curves $C_\pm$ do not bound spectrally unstable states.

In Figure \ref{fig:c} we depict dispersion relation \eqref{eq:disprel} as a function of the wavenumber, more precisely $\smash{\frac{c}{U_\sigma}}$ vs $\kappa$, for selected values of the various basic state parameters. The onset of instability happens at the wavenumber where the two branches of the dispersion relation merge. Included in the plot for reference is the asymptotic expression for $\smash{\frac{c}{U_\sigma}}$ as $\kappa\uparrow\infty$, as it follows from \eqref{eq:disprel-ass}.  Also indicated is $\smash{\lim_{\kappa\uparrow\infty}}\smash{\frac{c}{U_\sigma}} = \alpha$, that is, a real number, irrespective of the parameters that define the basic state. 

\begin{figure}
    \centering
    \includegraphics[width=.5\textwidth]{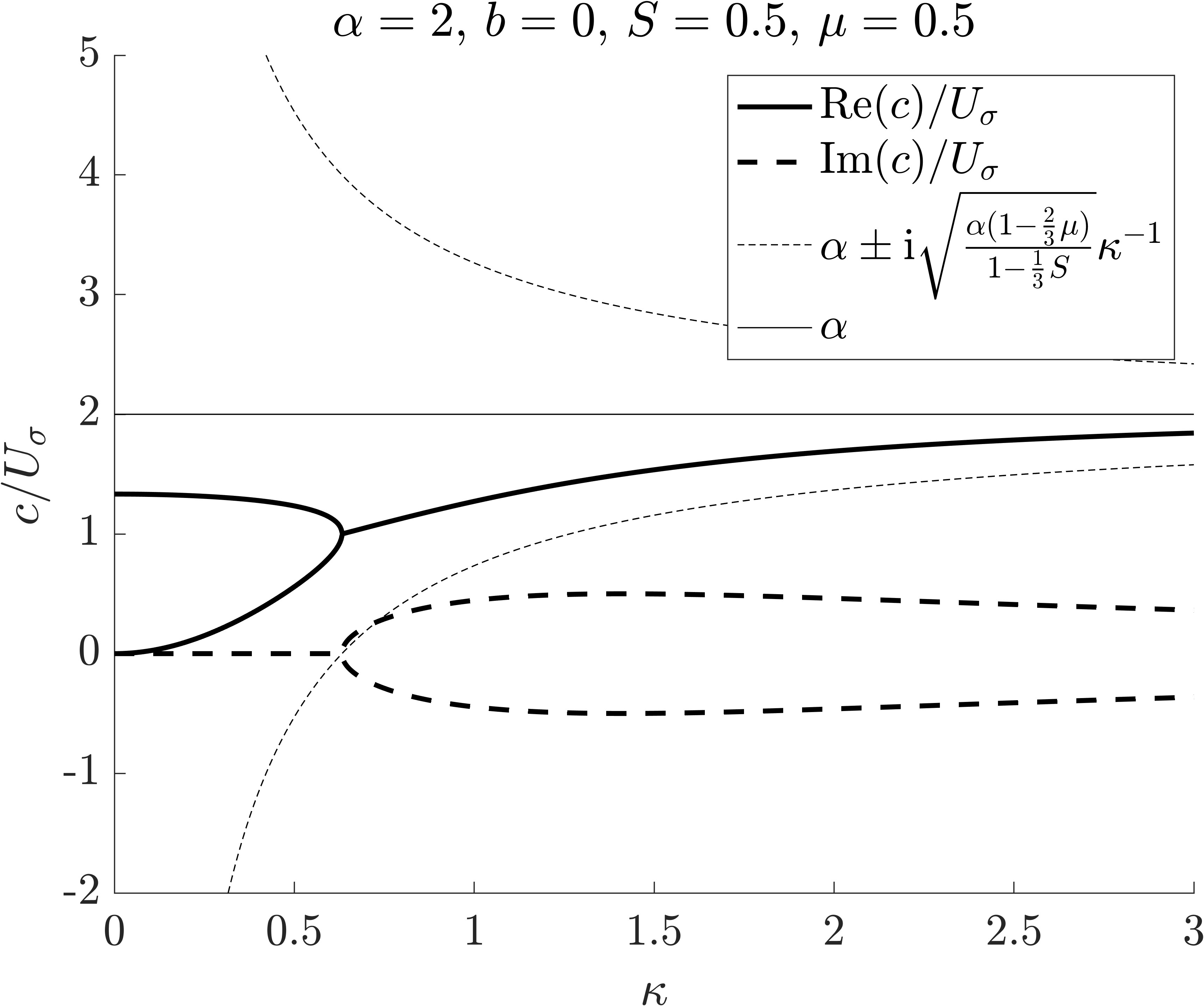}
    \caption{As a function of normalized wavenumber, normalized phase speed for a normal-mode perturbation on the baroclinic zonal jet defined by \eqref{eq:BS} with parameters as indicated. Asymptotic dispersion relation curves as $\kappa\uparrow\infty$ and the corresponding limiting value are included.}
    \label{fig:c}
\end{figure}

In the limit of weak stratification, i.e., $S\downarrow 0$, the length scales $R$ and $L := \smash{\frac{N_\mathrm{r}H_\mathrm{r}}{|f_0|}}$ are well separated.  The latter is proportional to the gravest baroclinic (internal) Rossby radius of deformation in a model with arbitrary stratification.  Since we are assuming a reduced-gravity setting, $R$ and $L$ can be approximately identified with the first and second internal deformation radii, respectively, of the arbitrarily stratified model extending throughout the entire water column. The noted scale separation allows one to distinguish long perturbations, with wavenumbers $\kappa = O(1)$ as $S\downarrow 0$, from short perturbations, with wavenumbers $L|\mathbf k| = \sqrt{2S}\kappa = O(1)$ as $S\downarrow 0$. We find that \eqref{eq:disprel} reduces to
\begin{equation}
    \lambda = \frac{-\nu(0) \pm \sqrt{\nu(0)^2 - 4\alpha\left(1 - \tfrac{2}{3}\mu\right)\left(\kappa^2 + 1\right)}}{2 \left(\kappa^2 + 1\right)},
    \label{eq:disprel-long}
\end{equation}
for long perturbations, and $\lambda = 0$, for short perturbations.  The limit $\kappa\uparrow \infty$ of \eqref{eq:disprel-long} gives $\lambda = 0$, i.e., $c = \bar U$.  As expected, it coincides with the short-perturbation phase speed.  This is real for all wavenumbers, which can be anticipated to be consequential for direct, fully nonlinear simulations.  It seems reasonable to think that wave activity in such simulations will fall off at sufficiently large wavenumbers (short wavelengths), which, according to the spectral (i.e., linear) analysis, are not growing.  

\begin{figure}
    \centering
    \raisebox{.5cm}{\includegraphics[width=.475\textwidth]{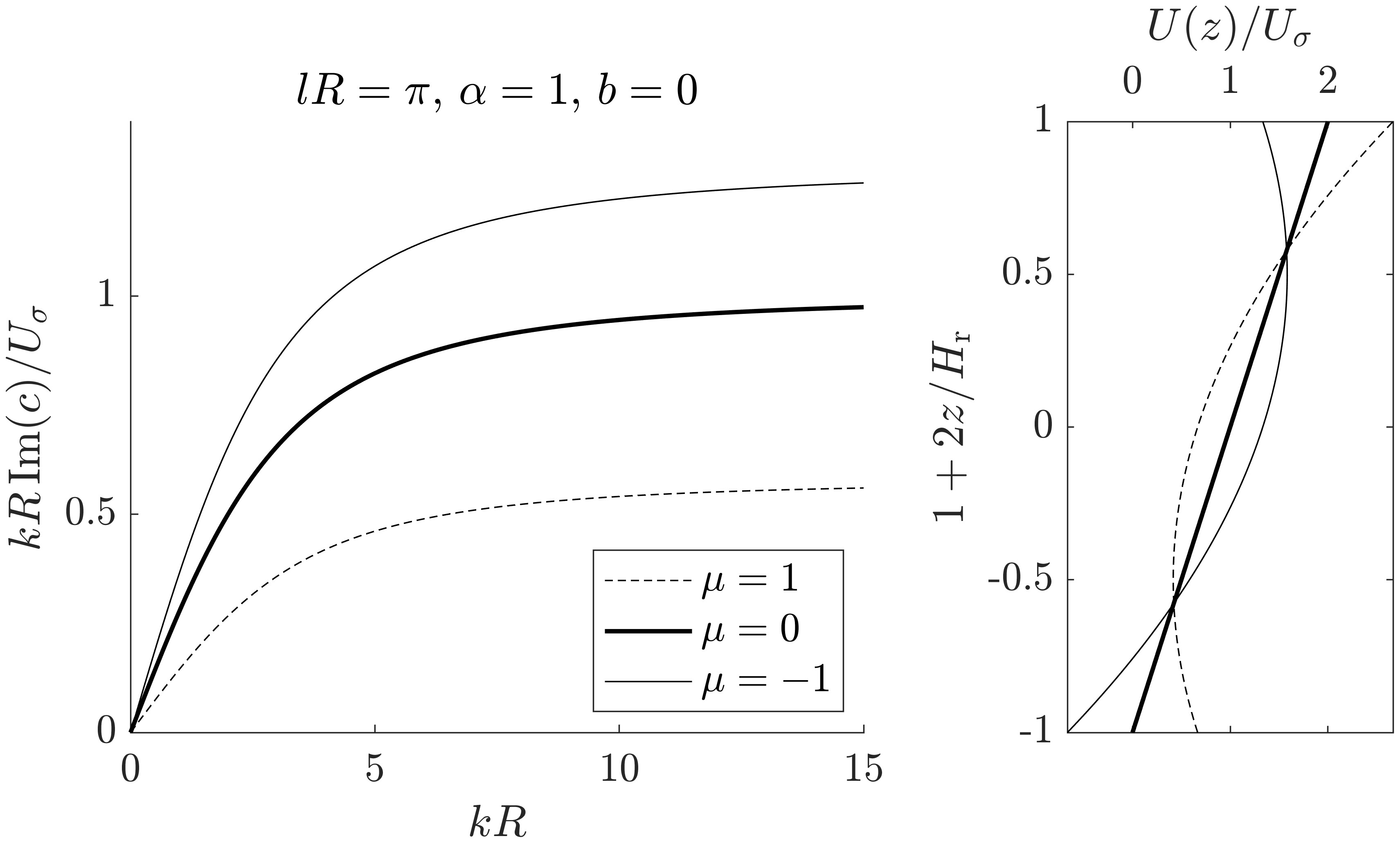}}\quad
    \includegraphics[width=.475\textwidth]{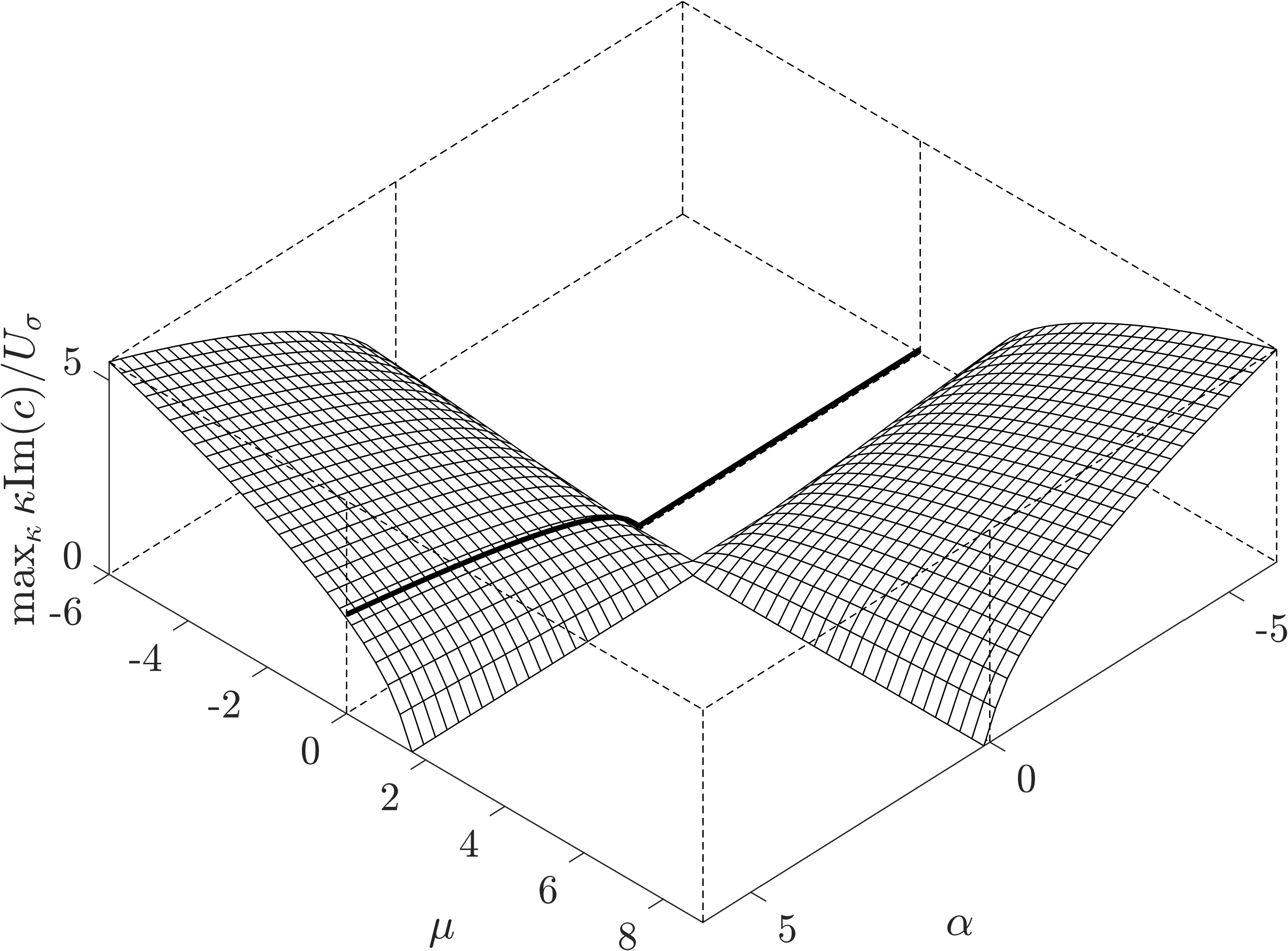}
    \caption{(right) Growth rate as a function of zonal wavenumber in the limit of weak stratification  ($S\downarrow 0)$ with basic state parameters $\alpha$ and $b$ as indicated for three values of $\mu$.  The IL$^0$QG result corresponds to the $\mu = 0$ curve as the velocity in that model can only include linear shear (implicitly, by the thermal-wind balance).  (middle) Zonal velocity (implicit) vertical profiles leading to the growth rates in the left panel. (left) Upper bound on the growth rate in $(\alpha,\mu)$-space for weak stratification. The thick line is the corresponding result for the IL$^0$Q.}
    \label{fig:omega}
\end{figure}

Concerning the growth rate of the perturbation, given by $\varpi := k\operatorname{Im}(c)$ for the smallest meridional wavenumber, viz., $l = \smash{\frac{\pi}{W}}$, it is observed that this saturates with zonal wavenumber, both when the stratification is finite and for long perturbations in the weak-stratification limit.  The left panel of Figure \ref{fig:omega} shows $\varpi(k)$ when $W = R$ for selected values of basic state parameters outside the spectrally stable set \eqref{eq:stab-spec-set} in the limit of weak stratification and with vanishing Charney number.  Note that $\varpi$ is larger (resp., smaller) when the basic state zonal velocity vertical profile (Figure \ref{fig:omega}, middle panel) curves eastward (resp., westward) than when it does not curve, as in the IL$^0$QG.  An upper-bound measure of the saturating value of $\varpi$ is given, when $S$ is finite, by
\begin{equation}
    \max_\kappa\frac{\kappa\operatorname{Im}(c)}{U_\sigma} = \sqrt{\frac{\alpha(1 - \tfrac{2}{3}\mu)}{1 - \tfrac{1}{3}S}},
    \label{eq:omega}
\end{equation}
where parameters $(\alpha,\mu)$ belong in the spectrally unstable region, i.e., the complement of set \eqref{eq:stab-spec-set}. The long-perturbation limit of \eqref{eq:omega} is $\smash{\sqrt{\alpha(1-\frac{2}{3}\mu})}$, which we plot in right panel of Figure \ref{fig:omega}.  An important final observation, with consequences for direct simulations, is that the growth rate for short perturbations vanishes for all wavenumbers, reinforcing the comments above that development of wave activity should be stopped in the IL$^{(0,1)}$QG at sufficiently short wavelengths.

\subsection{Formal stability}\label{sec:formal}

We now turn to investigate the formal stability of basic satate \eqref{eq:BS}, for which we have at our disposal several conservation laws, namely, energy ($\mathcal E$), Casimirs ($\mathcal C_{a,F}$), zonal momentum ($\mathcal M$), and the non-Casimir, non-explicit-symmetry-related conservation laws ($\mathcal I_F$) (cf.\ Section \ref{sec:integrals}) to construct an appropriate general integral of motion $\mathcal H$ such that $\delta\mathcal H = 0$ at a basic state. The use of energy and Casimirs to construct $\mathcal H$ may be traced from the work of Arnold in the 1960's \cite[cf.][which is a collection of papers that appeared in those years]{Arnold-14} back over a decade to those of \cite{Fjortoft-50} and \cite{Kruskal-Oberman-58}. The use of momentum in systems with $x$-translational symmetry was pioneered by Ripa \cite{Ripa-JFM-83}.  Here we explore the additional use of the weak Casimir conservation law $\mathcal I_F$.  We these comments in mind, we will, following tradition, refer to the use of all available integrals of motion to derive a-priori stability statements as \emph{Arnold's method}.

\begin{remark}\label{rem:formal}
    A few remarks are in order:  
    \begin{enumerate}
    \item The equilibria $\Phi$ of a Hamiltonian system $\partial_t\varphi = \{\varphi,\mathcal E\}$ are no longer unrestricted conditional extrema of the Hamiltonian ($\mathcal E$); by contrast, $\partial_t\Phi = \{\Phi,\mathcal E[\Phi]\}[\Phi] = 0$ implies that $\smash{\frac{\delta}{\delta\varphi}}(\mathcal E + \mathcal C)[\Phi] = 0$ for some Casimir $\mathcal C$. Level sets of constants of the motion define certain ``leaves'' on the Poisson manifold, $(P,\{\,,\hspace{.01cm}\})$. If these constants are the Casimirs of $\{\,,\hspace{.01cm}\}$, they form the ``symplectic leaves'' of $P$ \textup{\cite[e.g.,][]{Holm-etal-85,Morrison-98}}. Thus equilibria (basic states) are critical ``points'' of $\mathcal E$ restricted to such leaves.
    
    \item When $\mathcal H$ is chosen to be a linear combination $\mathcal E + \mathcal C - U\mathcal M$ where $U = \const$ and for some $\mathcal C$ such that $\delta\mathcal H = 0$ at $\Phi$, it turns out that $\delta^2\mathcal H$ is a Hamiltonian for the linearized dynamics about $\Phi$ as seen from a reference frame steadily translating in the $x$-direction with speed $U$. Namely, to the lowest order in $\delta\varphi$, it follows that $(\partial_t+U\partial_x)\delta\varphi = \{\delta\varphi,\delta^2\mathcal H\}[\Phi]$ where the bracket here is that defined in \eqref{eq:LP} but with a constant argument.  This is a genuine Lie--Poisson bracket in that all required properties are satisfied, most importantly the Jacobi identity, which is readily verified \textup{\cite[cf.,][for a discussion]{Morrison-Eliezer-86}}.
    
    \item With $\mathcal H$ as constructed, the integral $\delta^2\mathcal H$, or more generally $\Delta\mathcal H$, is referred to as \textbf{pseudoenergy--momentum}, following notation introduced in \textup{\cite{Ripa-JFM-83}}.  In the case with no zonal symmetry, the \emph{pseudoenergy} is referred to as a ``free energy'' in \textup{\cite{Morrison-98}} to mean that it is the energy accessible to the system upon perturbation away from equilibrium given the Casimir.  It should not be confused with the free energy defined in Section \ref{sec:waves}.
    \end{enumerate}
\end{remark}

We begin by considering
\begin{equation}
    \mathcal H_{\bar U} := \mathcal E - \bar U\mathcal M + \mathcal C_{0,F}
    \label{eq:H-C}
\end{equation} 
with
\begin{equation}
    F(\psi_\sigma,\psi_{\sigma^2}) = - \tfrac{1}{2} \alpha R_S^{-2}\left(\psi_\sigma^2 - \tfrac{2}{3}\mu^{-1} \psi_{\sigma^2}^2\right).
\end{equation}
The first variation of \eqref{eq:H-C} identically vanishes at the basic state \eqref{eq:BS}.  Its second variation
\begin{equation}
    \delta^2\mathcal H_{\bar U} = \mathcal E[\delta\bar\psi] - \tfrac{1}{2}R_S^{-2}\alpha\int 
    \delta\psi_\sigma^2 - \tfrac{2}{3}\mu^{-1} \delta\psi_{\sigma^2}^2
    \,d^2x,
    \label{eq:d2HU}
\end{equation}
which is positive definite provided that
\begin{equation}
    \alpha < 0,\quad \mu < 0.
    \label{eq:d2HUpos}
\end{equation}
This corresponds to the hatched region of the $(\alpha,\mu)$-plane of Figure \ref{fig:alphamu}.  That is, only a subset of the spectrally stable states are possible to be proved formally stable.  But such a subset however is stable with respect to perturbations not only of arbitrary structure.  Rather, they are stable with respect to finite-amplitude perturbations since
\begin{equation}
    \delta^2\mathcal H_{\bar U} = \Delta \mathcal H_{\bar U}. 
    \label{eq:d2HUequivDHU}
\end{equation}
That is, the pseudoenergy--momentum is a quadratic integral of motion.  This is a Hamiltonian for the exact dynamics of the IL$^{(0,1)}$ as seen by an observer zonally moving with speed $\bar U$ (extending Remark \ref{rem:formal}(ii), above). We will return to discussing the consequences of \eqref{eq:d2HUequivDHU} in Section \ref{sec:lyapunov}, below.

Let us consider now the use of the weak Casimir integrals $\mathcal I_F$, defined in \eqref{eq:I}, in Arnold's method. It turns out that the general integral of motion
\begin{equation}
    \mathcal H_{\tfrac{2}{3}\alpha U_\sigma} := \mathcal E - \tfrac{2}{3}\alpha U_\sigma\mathcal M + \mathcal C_{0,F} + \mathcal I_G
    \label{eq:HI}
\end{equation}
for
\begin{align}
     F(\psi_\sigma,\psi_{\sigma^2}) &=  \tfrac{1}{2}R_S^{-2}\alpha(\nu+1)\left(\tfrac{2}{3}\nu^{-1}\psi_{\sigma^2}^2 - \psi_\sigma^2\right),\\
     G\big(\psi_\sigma - \tfrac{2}{3}\psi_{\sigma^2}\big) &= \alpha \big(\psi_\sigma - \tfrac{2}{3}\psi_{\sigma^2}\big),
\end{align}
has a vanishing first variation at the basic state \eqref{eq:BS}. Its second variation,
\begin{align}
    \delta^2\mathcal H_{\tfrac{2}{3}\alpha U_\sigma}
     = {} & \mathcal E[\delta\bar\psi]
     +
     \alpha\int 
     \left(\delta\bar\xi + \tfrac{1}{2}\delta\psi_\sigma\right)^2 - 
     \left(\delta\bar\xi + \tfrac{1}{3}\delta\psi_{\sigma^2}\right)^2\nonumber\\
     & - \tfrac{1}{2}\left(R_S^{-2}(\nu+1) + \tfrac{1}{2}\right) \delta\psi_\sigma^2 + 
     \tfrac{1}{3}\left(R_S^{-2}\mu^{-1}(\nu+1) - \tfrac{1}{3}\right) \delta\psi_{\sigma^2}^2
     \,d^2x.
     \label{eq:d2HI}
\end{align}
Higher-order variations of \eqref{eq:HI} vanish, so \eqref{eq:d2HI} is an exact, fully nonlinear conservation law.  However, it is at most positive semidefinite and provided that $\alpha = 0$. Thus the conservation laws \eqref{eq:I} are not useful to make an \emph{a priori} assessment about the stability of the basic state \eqref{eq:BS}, at least using Arnold's method.

\subsection{Lyapunov stability}\label{sec:lyapunov}

Lyapunov stability of the basic states defined by \eqref{eq:d2HUpos} can be established as follows.  Consider the subset of stable basic states with
\begin{equation}
    \alpha < 0,\quad
    0 > \mu \ge -1. 
    \label{eq:d2HUpos-1}
\end{equation}
Let $\eta_1$ and $\lambda_1$ be two positive constants such that
\begin{equation}
    0 < \eta_1 < -\alpha
    \label{eq:eta-1}
\end{equation}
and
\begin{equation}
    -\alpha - \eta_1 \le \lambda_1 \le -\alpha + \eta_1,
    \label{eq:lambda-1}
\end{equation}
respectively. Recalling that \eqref{eq:d2HUequivDHU} holds, the convexity estimate follows:
\begin{equation}
   \frac{-\alpha-\eta_1}{\lambda_1} \|(\delta\bar\xi,\delta\psi_\sigma,\delta\psi_{\sigma^2})\|_2(\lambda_1)^2 \le \Delta \mathcal H_{\bar U} \le \frac{-\alpha+\eta_1}{\lambda_1} \|(\delta\bar\xi,\delta\psi_\sigma,\delta\psi_{\sigma^2})\|_2(\lambda_1)^2,
   \label{eq:conv-1}
\end{equation}
where  
\begin{equation}
    \|(\delta\bar\xi,\delta\psi_\sigma,\delta\psi_{\sigma^2})\|_2(\lambda_1)^2 := \mathcal E[\delta\bar\psi] + \tfrac{1}{2}R_S^{-2}\lambda_1\int \delta\psi_\sigma^2 + \tfrac{2}{3} \delta\psi_{\sigma^2}^2
    \,d^2x,
    \label{eq:L2-1}
\end{equation}
measures, in an $L^2$ sense, the squared distance to the subfamily of basic states \eqref{eq:BS} defined by \eqref{eq:d2HUpos-1} in the infinite-dimensional phase space of the IL$^{(0,1)}$QG model equation \eqref{eq:IL01} with coordinates $(\bar\xi,\psi_\sigma,\psi_{\sigma^2})$. This establishes Lyapunov stability for them.  Indeed, since $\Delta \mathcal H_{\bar U}$ is invariant and convex, it follows that the distance to such states at any time is bounded from above by a multiple of the initial distance.

Consider now the subfamily of basic states \eqref{eq:BS} defined by
\begin{equation}
    \alpha < 0,\quad
    \mu < -1. 
    \label{eq:d2HUpos-2}
\end{equation}
Letting $\eta_2$ and $\lambda_2$ be positive constants such that
\begin{equation}
    0 < \eta_2 < \frac{\alpha}{\mu}
    \label{eq:eta-2}
\end{equation}
and
\begin{equation}    
    \frac{\alpha}{\mu} - \eta_2 \le \lambda_2 \le \frac{\alpha}{\mu} + \eta_2,
    \label{eq:lambda-2}
\end{equation}
one finds the convexity estimate:
\begin{equation}
   \frac{\alpha-\mu\eta_2}{\mu\lambda_2} \|(\delta\bar\xi,\delta\psi_\sigma,\delta\psi_{\sigma^2})\|_2(\lambda_2)^2 \le \Delta \mathcal H_{\bar U} \le \frac{\alpha+\mu\eta_1}{\mu\lambda_2} \|(\delta\bar\xi,\delta\psi_\sigma,\delta\psi_{\sigma^2})\|_2(\lambda_2)^2.
   \label{eq:conv-2}
\end{equation}
This establishes Lyapunov stability for basic states \eqref{eq:BS} with \eqref{eq:d2HUpos-2}. Overall the above convexity estimates establish Lyapunov stability for \eqref{eq:BS} over the whole range of stable parameters \eqref{eq:d2HUpos}.

\section{Nonlinear saturation of unstable baroclinic zonal flow}\label{sec:bounds}

In this section we seek to a-prior constraining the growth of perturbations to basic states \eqref{eq:BS} that have been found to be spectrally unstable, namely, those that violate \eqref{eq:stab-spec}, by making use of the above Lyapunov stability result(s) for basic states \eqref{eq:BS} satisfying \eqref{eq:d2HUpos}.  This can be done using \emph{Shepherd's method}, introduced in \cite{Shepherd-88a}, which we outline below.  Denote stable (resp., unstable) state quantities with a superscript $\text{S}$ (resp., $\text{U}$).  Grouping $(\bar\xi,\psi_\sigma,\psi_{\sigma^2})$ into $\varphi$ we have
\begin{equation}
    \|\varphi - \Phi^\text{U}\|_2(\lambda_1)\le \sqrt{\frac{2\alpha^\text{S}}{\alpha^\text{S}+\eta_1}} \|\Phi^\text{S} - \Phi^\text{U}\|_2(\lambda_1) =: \mathcal B_1
    \label{eq:B1}
\end{equation}
for $\Phi^\text{S}$ satisfying \eqref{eq:d2HUpos-1} and
\begin{equation}
    \|\varphi - \Phi^\text{U}\|_2(\lambda_2)\le \sqrt{\frac{2\alpha^\text{S}}{\alpha^\text{S}-\mu^\text{S}\eta_2}} \|\Phi^\text{S} - \Phi^\text{U}\|_2(\lambda_2) =: \mathcal B_2
    \label{eq:B2}
\end{equation}
for $\Phi^\text{S}$ satisfying \eqref{eq:d2HUpos-2}. The inequalities above follow as a result of successively:
\begin{enumerate}
    \item applying the triangular inequality $\|\varphi - \Phi^\text{U}\|_2(\lambda_i) \le \|\varphi - \Phi^\text{S}\|_2(\lambda_i) + \|\Phi^\text{S} - \Phi^\text{U}\|_2(\lambda_i)$, $i=1,2$;
    \item using convexity estimate \eqref{eq:conv-1} for $i=1$ and \eqref{eq:conv-2} for $i=2$; and 
    \item \emph{assuming that $\varphi \approx \Phi^\text{U}$ initially at $t = 0$}.
\end{enumerate} 

Set $\lambda_1$ and $\lambda_2$ to their minimum values in their corresponding admissible ranges, given by \eqref{eq:lambda-1} and \eqref{eq:lambda-2}, respectively. Let
\begin{equation}
    0<r_i<1,\quad i = 1,2,
\end{equation}
be constants such that
\begin{equation}
    \eta_1 = - r_1\alpha^\text{S},\quad 
    \eta_2 = r_2\frac{\alpha^\text{S}}{\mu^\text{S}}
\end{equation}
satisfy \eqref{eq:eta-1} and \eqref{eq:eta-2}, respectively.  The following tighter bounds result upon minimizing over $r_i$, $i=1,2$:
\begin{equation}
   \mathcal B_1 \ge \hat{\mathcal B_1} := R_S^{-1}\sqrt{\tfrac{1}{6}LW^3}\cdot\sqrt{(\bar U^\text{S} -\bar U^\text{U})^2 - \frac{\bar U^\text{S}}{U_\sigma^\text{S}}\left((U_\sigma^\text{S} - U_\sigma^\text{U})^2 + \tfrac{2}{3}(U_{\sigma^2}^\text{S} - U_{\sigma^2}^\text{U})^2\right)}
   \label{eq:hatB1}
\end{equation}
and
\begin{equation}
   \mathcal B_2 \ge \hat{\mathcal B_2} := R_S^{-1}\sqrt{\tfrac{1}{6}LW^3}\cdot\sqrt{(\bar U^\text{S} -\bar U^\text{U})^2 + \frac{\bar U^\text{S}}{U_{\sigma^2}^\text{S}}\left((U_\sigma^\text{S} - U_\sigma^\text{U})^2 + \tfrac{2}{3}(U_{\sigma^2}^\text{S} - U_{\sigma^2}^\text{U})^2\right)}.
   \label{eq:hatB2}
\end{equation}

\change[RA]{In the top-left panel of \mbox{Fig.\ \ref{fig:bounds}} we plot $\hat{\mathcal B}_1$ in $(\alpha^\text{U},\mu^\text{U})$-space for $U_\sigma^\text{U} = 2f_0R$ as computed using $\smash{\bar U^\text{S}} = f_0R$, $U_\sigma^\text{S} = -f_0R$, and $\smash{U_{\sigma^2}^\text{S}} = \smash{\frac{1}{2}}f_0R$. The top-right panel of  \mbox{Fig.\ \ref{fig:bounds}} likewise shows $\smash{\hat{\mathcal B}_2}$ as obtained with  $\bar U^\text{S} = f_0R$, $\smash{U_\sigma^\text{S}} = -f_0R$, and $\smash{U_{\sigma^2}^\text{S}} = \smash{\frac{3}{2}}f_0R$.  Note that the bounds do not drop down to zero at the boundary of instability.  They do along $\alpha^\text{U} = 0$ as $\smash{\bar U^\text{S} \downarrow 0}$ (\mbox{Fig.\ \ref{fig:bounds}}, bottom panels). However, they do not along $\smash{\mu^\text{U} = \frac{3}{2}}$ because the Lyapunov-stable basic states only span a subset of the complement of the spectrally stable region in the $(\alpha,\mu)$-space (\mbox{cf.\ Fig.\ \ref{fig:alphamu}}).  Optimal bounds may be obtained by minimizing $\smash{\hat{\mathcal B}_1}$ and $\smash{\hat{\mathcal B}_2}$ over stable state parameters. The important point to highlight however is the existence of a priori upper bounds which prevent perturbations to spectrally unstable basic states, namely, equilibria \mbox{\eqref{eq:BS}} violating \mbox{\eqref{eq:stab-spec}}, from undergoing an ultraviolet explosion.}{In the upper-left panel of \mbox{Fig.\ \ref{fig:bounds}}, we illustrate $\hat{\mathcal B}_1$ within the $(\alpha^\text{U},\mu^\text{U})$-space for $U_\sigma^\text{U} = 2f_0R$, calculated utilizing $\smash{\bar U^\text{S}} = f_0R$, $U_\sigma^\text{S} = -f_0R$, and $\smash{U_{\sigma^2}^\text{S}} = \smash{\frac{1}{2}}f_0R$. Similarly, the upper-right panel of \mbox{Fig.\ \ref{fig:bounds}} depicts $\smash{\hat{\mathcal B}_2}$, derived using $\bar U^\text{S} = f_0R$, $\smash{U_\sigma^\text{S}} = -f_0R$, and $\smash{U_{\sigma^2}^\text{S}} = \smash{\frac{3}{2}}f_0R$. It is important to note that at the boundary of instability, these bounds do not reach zero. They indeed approach zero along $\alpha^\text{U} = 0$ as $\smash{\bar U^\text{S} \downarrow 0}$ (\mbox{Fig.\ \ref{fig:bounds}}, bottom panels), but not along $\smash{\mu^\text{U} = \frac{3}{2}}$, since the Lyapunov-stable basic states only cover a portion of the complement to the spectrally stable zone in the $(\alpha,\mu)$-space (\mbox{cf.\ Fig.\ \ref{fig:alphamu}}). Optimal bounds can be found by minimizing $\smash{\hat{\mathcal B}_1}$ and $\smash{\hat{\mathcal B}_2}$ over stable state parameters. The key point is the presence of a-priori upper bounds that restrict disturbances to spectrally unstable basic states, specifically equilibria \mbox{\eqref{eq:BS}} that violate \mbox{\eqref{eq:stab-spec}}, from experiencing an ultraviolet blowup.}  

\change[RA]{Finally, as shown in \mbox{\cite{Beron-21-POFa}}, Shepherd's method can be used to bound the growth of perturbations on unstable states which may belong to a class different than that of those states for which Lyapunov stability can be established. As a consequence, the above conclusion can be extended to perturbations on the vacuum state, characterized by positive-semidefinite free energy (\mbox{cf.\ Section \ref{sec:waves}}), which constitutes a special form of pseudo-energy--momentum.  Indeed, their nonlinear growth is constrained by the existence of Lyapunov-stable basic states as is the nonlinear growth of perturbations to unstable basic states, which have sign-indefinite pseudo-energy--momentum.  Bounds on the spontaneous growth of perturbations to the vacuum state, should it be realized, are given by either \mbox{\eqref{eq:hatB1}} and \mbox{\eqref{eq:hatB2}} with $\bar U^\text{U} = U_\sigma^\text{U} = U_{\sigma^2}^\text{U} = 0$.}{Finally, following the findings presented in \mbox{\cite{Beron-21-POFa}}, Shepherd's method serves as a means to constrain the growth of perturbations on unstable states that may pertain to a distinct class from those for which Lyapunov stability is established. Consequently, this inference extends to perturbed vacuum states, which are characterized by positive-semidefinite free energy (\mbox{cf.\ Section \ref{sec:waves}}), representing a variant of pseudo-energy--momentum. Essentially, their nonlinear growth is restricted by the presence of Lyapunov-stable basic states, similar to the nonlinear growth of perturbations to unstable basic states with sign-indefinite pseudo-energy--momentum. The constraints on the potential spontaneous growth of perturbations to the vacuum state, if it occurs, are given by either \mbox{\eqref{eq:hatB1}} or \mbox{\eqref{eq:hatB2}} with $\bar U^\text{U} = U_\sigma^\text{U} = U_{\sigma^2}^\text{U} = 0$.}

\begin{figure}
    \centering
    \includegraphics[width=.75\textwidth]{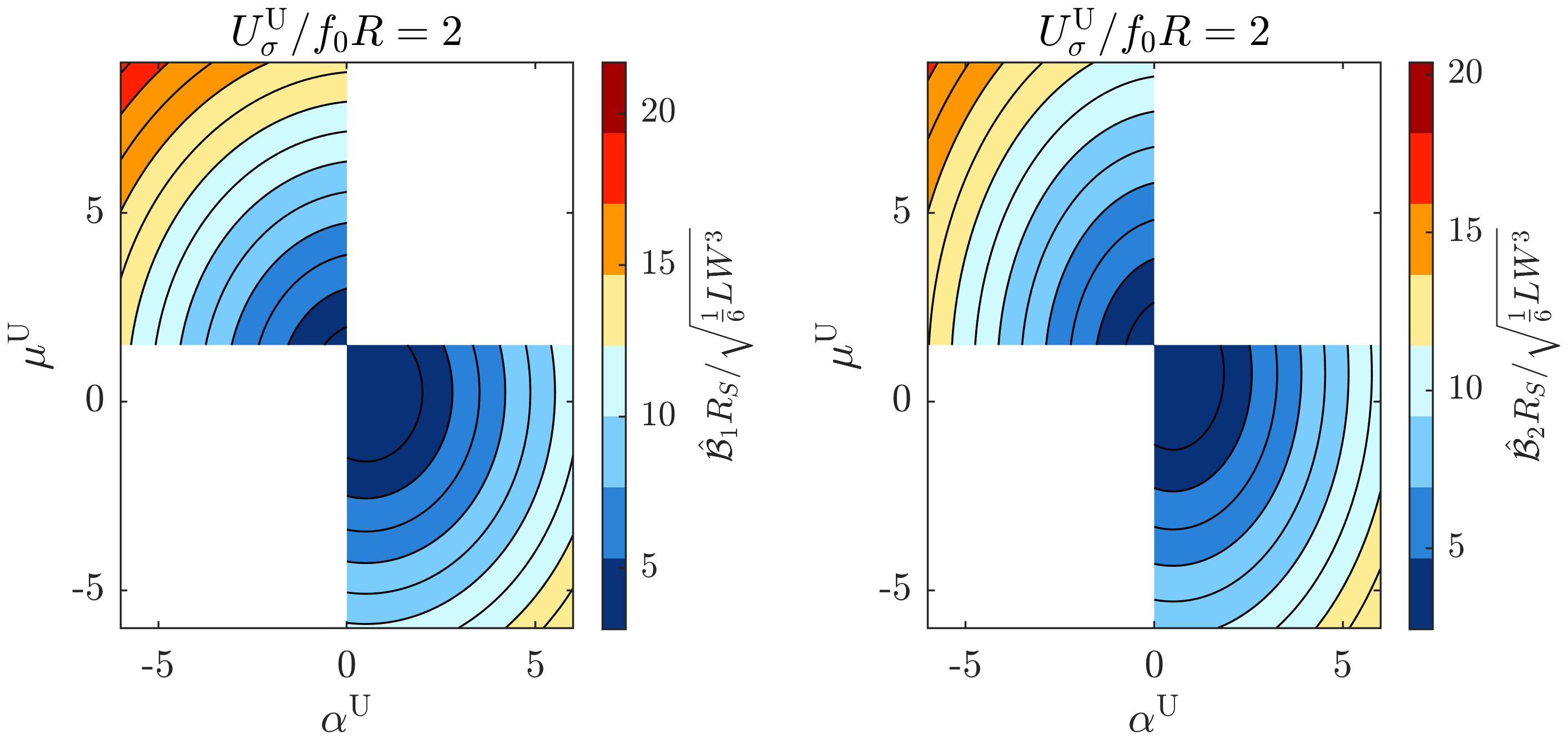}\\\vspace{.25cm}
    \includegraphics[width=.75\textwidth]{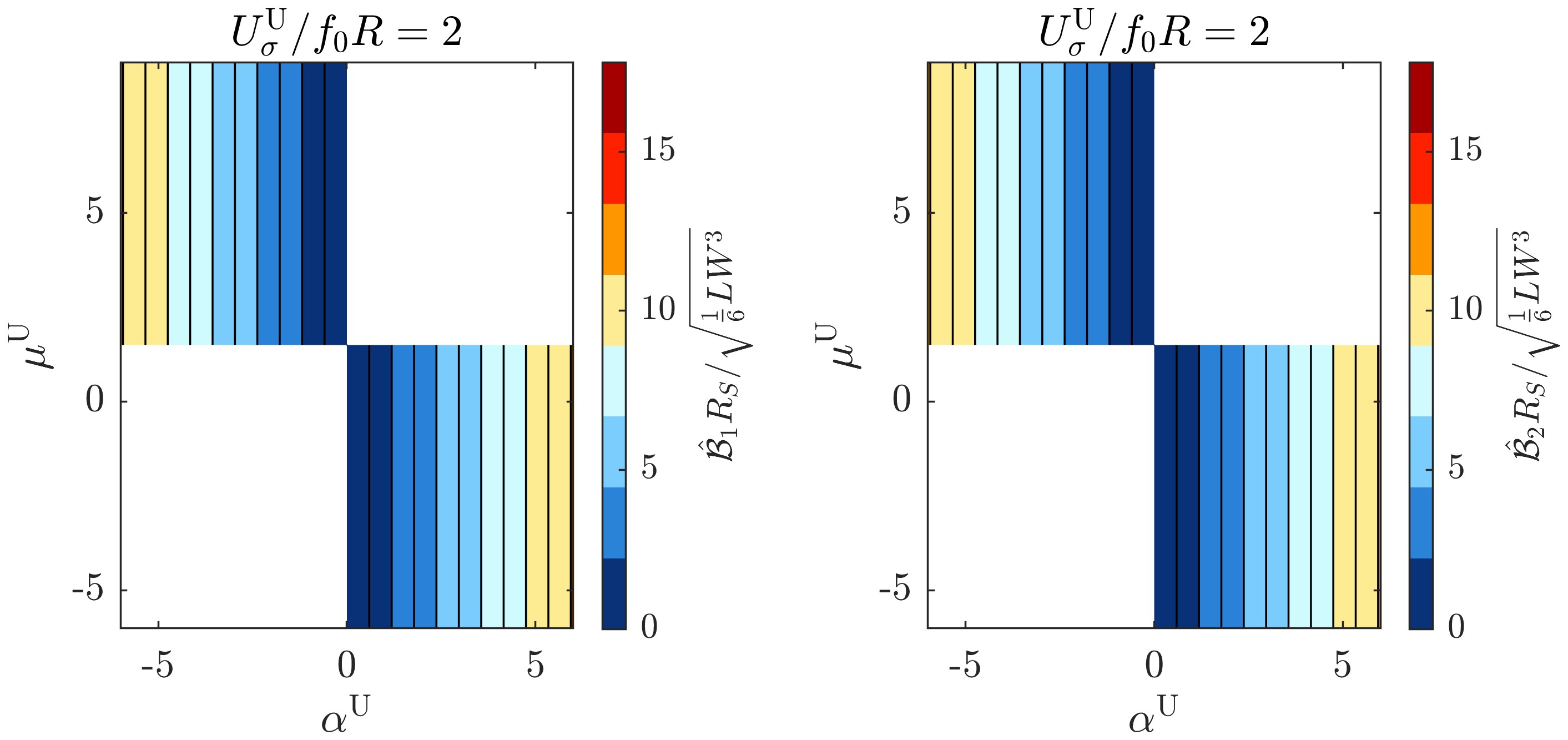}
    \caption{Normalized a-priori nonlinear upper bounds on the growth of perturbations to spectrally unstable basic states as computed using Lyapunov-stable states with $\smash{\alpha^\text{S}} < 0$ and $0 > \smash{\mu^\text{S}} \ge -1$ (top left) and $\smash{\alpha^\text{S}} < 0$ and $\smash{\mu^\text{S}} < -1$ (top right). The bottom panels show the same as in the top panels in the limit $\smash{\bar U^\text{S}} \downarrow 0$.}
    \label{fig:bounds}
\end{figure}

\section{Discussion}\label{sec:discussion}

The study presented here was motivated, in part, by direct numerical simulations revealing a tendency of the IL$^0$QG to quickly develop visually more intense small-scale vortex rolls than the IL$^{(0,1)}$QG \cite{Beron-21-POFb}.  Such small-scale circulation can be fairly referred to as submsoscales circulations as their scales fall well below $R$, the Rossby deformation radius in the IL$^0$QG and IL$^{(0,1)}$QG models.  As noted above, being $R$ an equivalent barotropic deformation radius, it approximately corresponds to the gravest baroclinic (i.e., internal) deformation scale in the continuously stratified model defined over the entire water column. The results from the various analyses carried out above do not appear to provide reason to expect that the numerical observation of \cite{Beron-21-POFb} holds in general, as we summarize next:

\begin{enumerate}
    \item  The production (or destruction) of Kelvin--Noether circulation ($\mathcal K$) along material loops was identified in \cite{Holm-etal-21} with the development of \change[]{sub-deformation}{subdefromation} scale or submesoscale wave activity in the IL$^0$QG.  It turns out that $\mathcal K$ is equally produced (or destroyed) in both the IL$^0$QG and IL$^{(0,1)}$QG.  The rate of change of $\mathcal K$ in the former is given by \eqref{eq:dotK} with the $-\smash{\frac{2}{3}\psi_{\sigma^2}}$ term omitted.  While by static stability one has that $\smash{-\frac{2}{3}\psi_{\sigma^2}} < \smash{\frac{1}{6}f_0R^2S}$, one cannot anticipate the contribution of the gradient of this term to the integral in \eqref{eq:dotK}.  In other words, one cannot anticipate the role of stratification in the production (or destruction) of $\mathcal K$.  
    
    \item Both the IL$^0$QG and IL$^{(0,1)}$QG models exhibit a neutral mode, termed a ``force compensating mode'' by \cite{Ripa-JGR-96}, where potential vorticity changes arising from buoyancy fluctuations do not alter the fluid velocity (cf.\ Section \ref{sec:waves}).  That is, associated with this mode is a vanishing free energy, which cannot constrain a possible spontaneous growth of this mode.  Thus, both the IL$^0$QG and IL$^{(0,1)}$QG suffer from this property, which might equally play a role in the development of submesoscale circulations in both models. 
    
    \item The intensity of submesoscale motions is constrained by the existence of Lyapunov-stable states, both in the IL$^{(0,1)}$QG (cf.\ Section \ref{sec:bounds}) and the IL$^0$QG \cite{Beron-21-POFa, Beron-24-POFa}.  However, one cannot compare the size of the space available in each model for wave activity, more precisely, the nonlinear growth of the amplitude of perturbations on basic states \eqref{eq:BS} violating \eqref{eq:stab-spec}, the condition for spectral stability.  The reason is that this is measured using different ($L^2$) norms.  
    
    \item When considering stability itself, it was found that this predicts growth rates in the IL$^{(0,1)}$QG that can be smaller than, equal to, or larger than those in the IL$^0$QG  (cf., e.g., Fig.\ \ref{fig:omega}).  The only, potentially important, difference is that the growth rate vanishes in the IL$^{(0,1)}$QG for short perturbations.  Distinguishing short from large perturbations in the weak stratification limit is only possible in IL$^{(0,1)}$QG. \change[RA]{This speaks about stability in the IL$^{(0,1)}$QG in high-wavenumber end of the spectrum and ensuing potentially less intense submesoscale wave activity in the nonlinear regime.}{This suggests stability in the IL$^{(0,1)}$QG at the high-wavenumber end of the spectrum, potentially resulting in less intense submesoscale wave activity in the nonlinear regime.}  
    
    \item Yet, that Arnold's method fails to demonstrate formal, let alone Lyapunov, stability for all spectrally stable states (cf.\ Fig.\ \ref{fig:alphamu}) suggests that the IL$^{(0,1)}$QG may be more prone to instability than the IL$^0$QG, in which case all spectrally stable basic states (in a similar class) are provable Lyapunov stable \cite{Beron-21-POFa, Beron-24-POFa}.
\end{enumerate}   

\begin{figure}
    \centering
    \includegraphics[width=.4\textwidth]{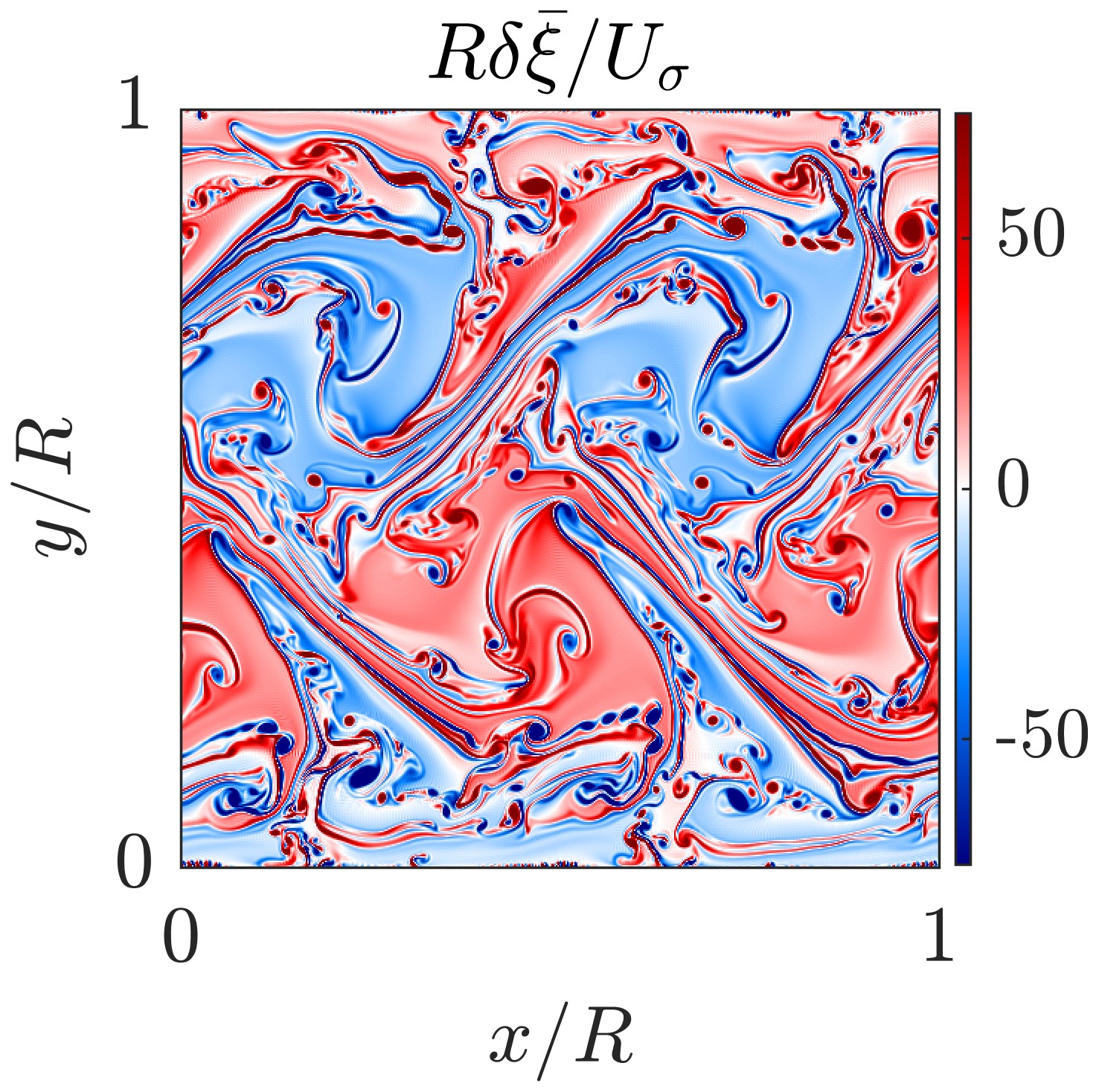}\,
    \includegraphics[width=.4\textwidth]{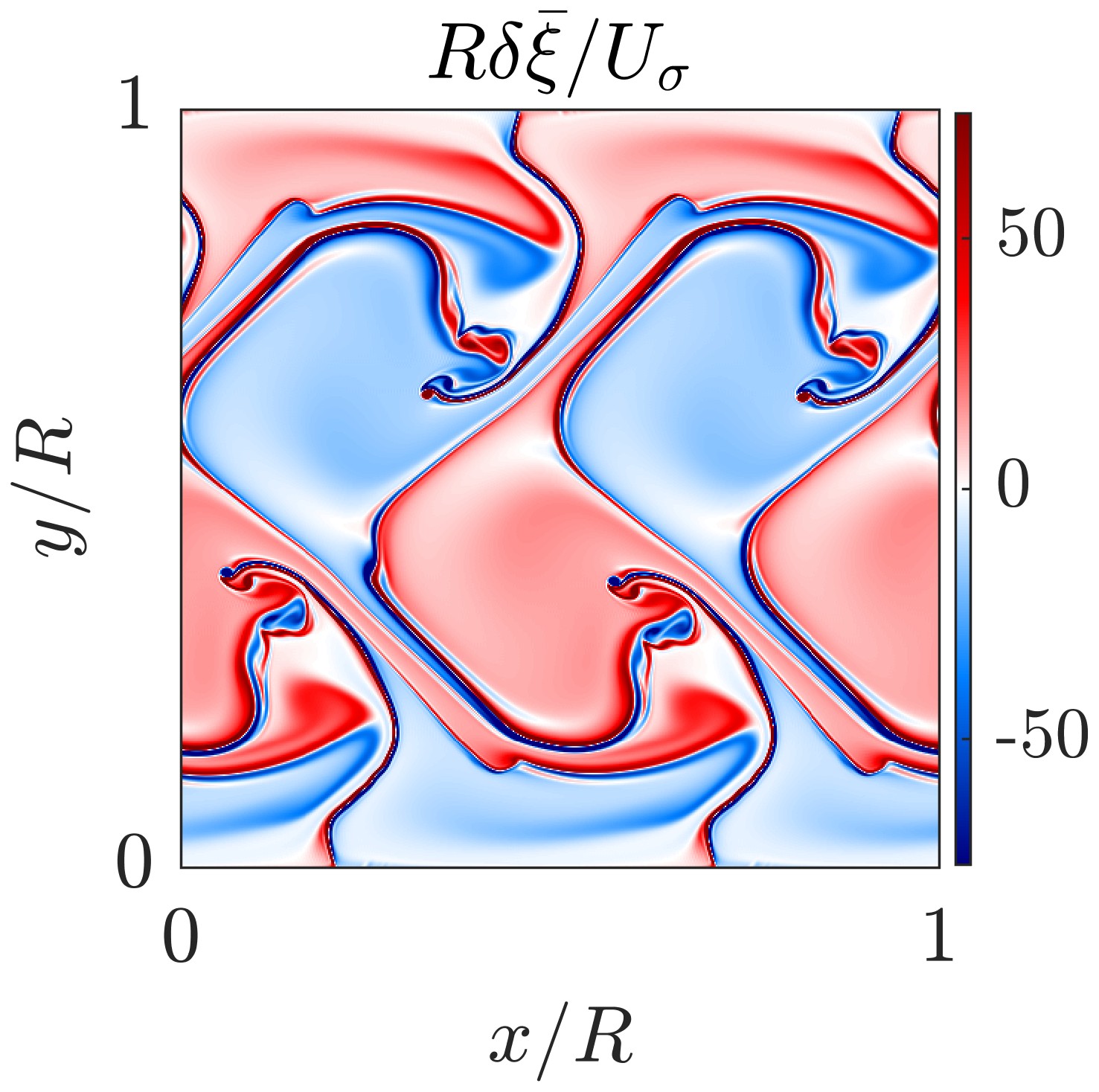}\\\vspace{.25cm}
    \includegraphics[width=.4\textwidth]{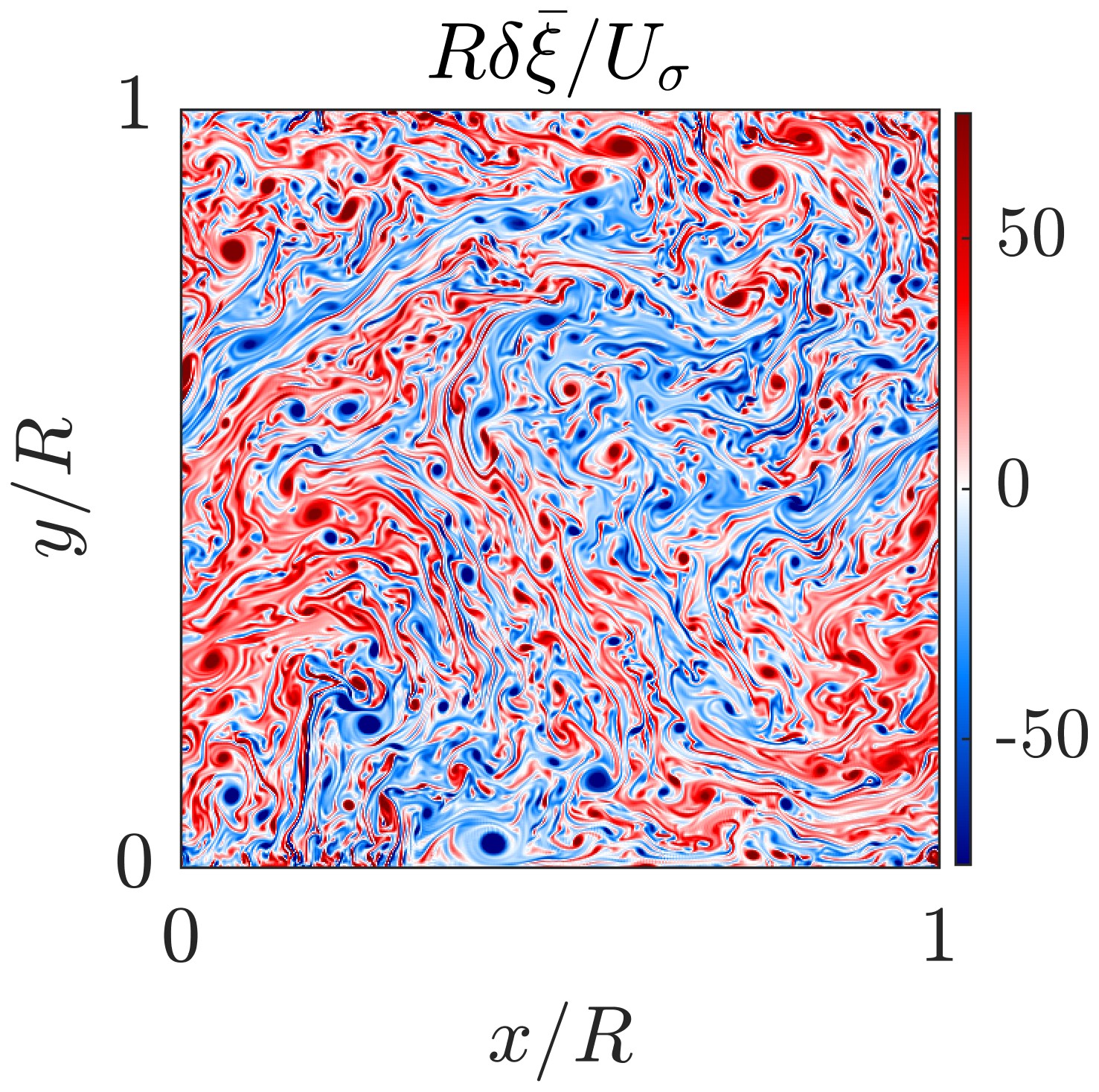}\,
    \includegraphics[width=.4\textwidth]{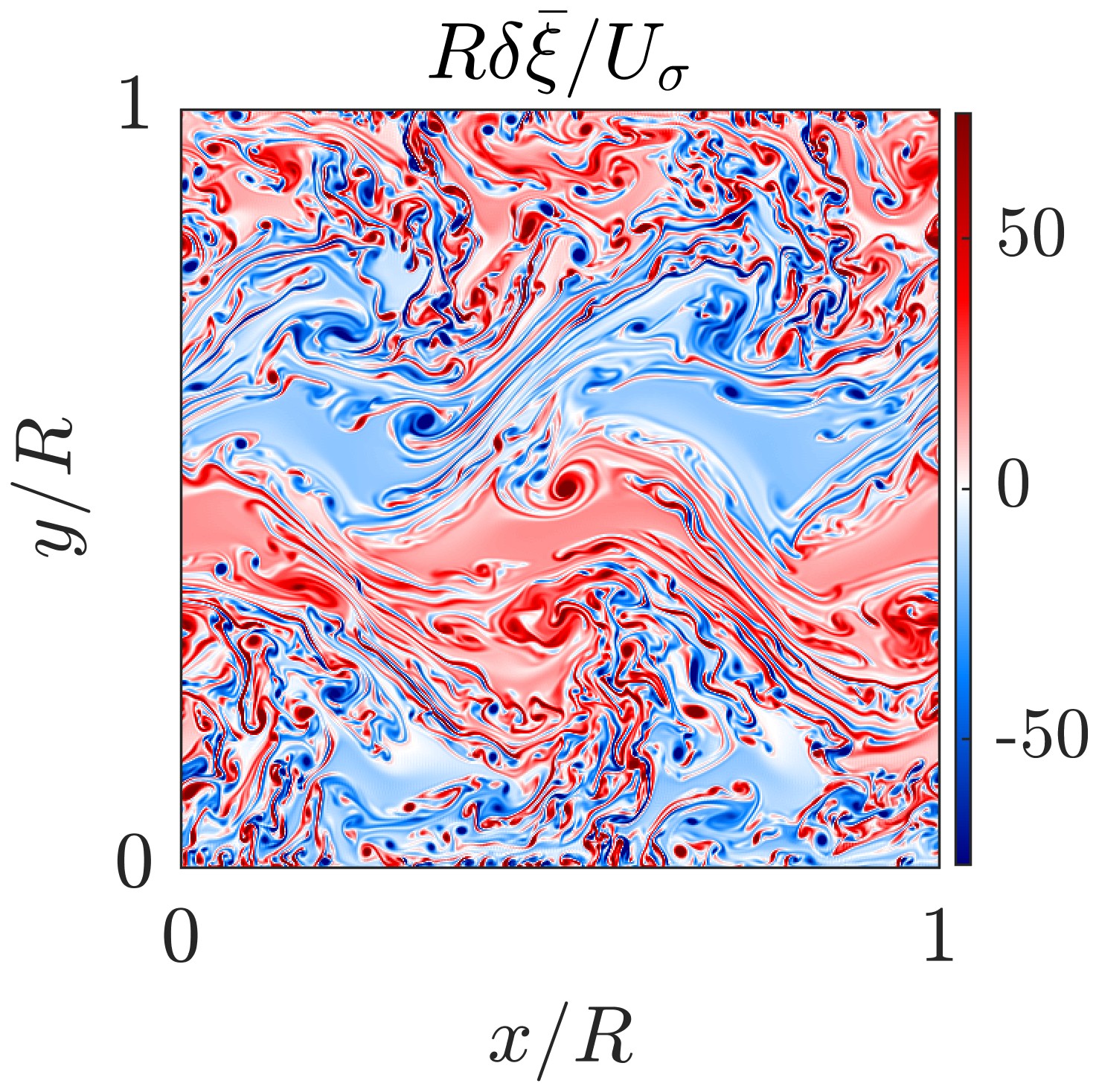}
    \caption{(left column) Snapshots at $U_\sigma t/R = 2$ (top) and 3 (bottom) of the nonlinear evolution of the potential vorticity perturbation to spectrally unstable basic state potential vorticity (\ref{eq:BS}d) with $\alpha = 3$, $\mu = 0$, $b = 0$, and $S = 0$. (right column) As in the left column, but based on the IL$^0$QG, for which $S = 0$ and $\mu = 0$ (the buoyancy field is vertically uniform in this model).}
    \label{fig:direct}
\end{figure}

We back up the above inferences with the results from direct numerical simulations (cf.\ Appendix for numerical details) of the IL$^{(0,1)}$ and IL$^0$QG.  These differ from those of \cite{Beron-21-POFb}, which included Hamiltoninan topographic forcing (cf.\ Section \ref{sec:topo}).  Concretely, we considered the full nonlinear evolution of perturbations to the spectrally unstable basic state \eqref{eq:BS} defined by $\alpha = 3$, $\mu = 0$, $b = 0$, and $S = 0$, which violates  \eqref{eq:stab-spec}. The initial perturbation was chosen to be a small-amplitude normal mode.  Taking the zonal channel domain period ($L$) and width ($W$) to be equal to the deformation radius ($R$), the initial perturbation reads 
\begin{equation}
    \frac{R\delta\bar\xi}{U_\sigma} =  -\frac{\delta\psi_\sigma}{RU_\sigma} = \frac{\delta\psi_{\sigma^2}}{RU_\sigma} = \frac{1}{2}\cos 4\pi \frac{x}{R}\sin \pi \frac{y}{R}.
\end{equation}
For wavenumber-(2,1) $\mathbf k = (k,l) = \smash{\big(\frac{4\pi}{R},\frac{\pi}{R}\big)}$ and the above basic state parameters, from dispersion relation \eqref{eq:disprel} we compute a growth rate of $\varpi \approx 1.7\smash{\frac{U_\sigma}{R}}$ for the initial perturbation.  The duration of the simulation was set to roughly $5\varpi^{-1}$, that is, five times the e-folding time of a growing normal mode.  Snapshots of perturbation potential vorticity, $\delta\bar\xi$, at $t = 2.5\varpi^{-1}$ and $5\varpi^{-1}$ are shown in the top- and bottom-left panel of Fig.\ \ref{fig:direct}, respectively.  Note the development of Kelvin--Helmoltz-like vortex rolls with scales much smaller than $R$.  Initially, a scale separation is very evident.  This evidence faints with time.  The right panels of Fig.\ \ref{fig:direct} show at corresponding times $\delta\bar\xi$ but in the IL$^0$QG.  This was initialized from the same basic state as the IL$^{(0,1)}$QG except that in the IL$^0$QG, $\smash{\psi_{\sigma^2} = 0}$, i.e., the buoyancy is vertically uniform, so $S = 0$ and $\mu = 0$.  For the parameters chosen and wavenumber-(2,1), the growth rate of the initial perturbation in the IL$^0$QG coincides with that in the IL$^{(0,1)}$QG.  Comparing the left and right columns of Fig.\ \ref{fig:direct}, it is clear that submesoscale wave activity in the IL$^{(0,1)}$QG can be more intense than, or at least as intense as, in the IL$^0$QG. 

\section{Concluding remarks}\label{sec:conclusions}

In this paper, we have investigated the properties of a thermal rotating shallow-water model with uniform stratification for subinertial upper-ocean dynamics, and carried out a stability analysis in the model. The model has an active layer, limited above by a rigid interface with the atmosphere, which freely floats atop an infinitely deep inert (abyssal) layer. The dynamics in the model is quasigeostrophic (QG).  Unlike the adiabatic reduced-gravity QG model, the thermal model allows buoyancy (temperature) to vary laterally (i.e., in the horizontal) and time.  This enables the incorporation of heat and freshwater fluxes across the ocean surface. By contrast with the standard thermal model, the model considered here includes uniform stratification, still maintaining the two-dimensional structure of the thermal and adiabatic (i.e., with a homogeneous density) models.  The model is a special case of a recently derived extended theory of thermal models with geometry \cite{Beron-21-POFb}.  The QG version of the extended family of models is referred to as IL$^{(0,n)}$QG.  Here IL stands for inhomogeneous layer. The first slot in the superscript of IL refers to the amount of vertical variation of the horizontal velocity with depth and the second one to the (whole) degree of a polynomial expansion of the buoyancy in the vertical coordinate.  The particular stratified thermal QG model considered here is the IL$^{(0,1)}$QG member of that family.  The standard thermal QG model \cite{Ripa-RMF-96, Warneford-Dellar-13, Holm-etal-21} is referred to as the IL$^0$QG, while the adiabatic QG model \cite{Pedlosky-87} to as the HLQG.

\subsection{Main findings}

The IL$^{(0,1)}$QG has unique solutions, it possess a Kelvin--Noether-like circulation theorem along material loops stating that creation (or annihilation) of circulation is due to the misalignment between the gradients of layer thickness and buoyancy, and is manifestly Hamiltonian with a Lie--Poisson bracket.  Geometrically, the bracket is a product for a realization Lie enveloping algebra on functionals in the dual of the Lie algebra of the Lie group obtained by extending the group of symplectic diffeomorphisms of the fluid domain, chosen here to be a periodic zonal channel of the $\beta$-plane,  by semidirect product with two copies of the space of smooth functions on that domain.  The induced algebra representation on the space of functions is given by the canonical Poisson bracket in $\mathbb R^2$.  This Hamiltonian structure is not spoiled by the addition of certain type of topographic forcing. The IL$^{(0,1)}$QG thus preserves energy and zonal momentum, respectively related to symmetry under time shits and zonal translations by Noether's theorem.  In addition to these conservation laws, the model has integrals of motion that form the kernel of the Lie--Poisson bracket, known as Casimirs.  Unlike the usual situation in Lie--Poisson Hamiltonian systems, the IL$^{(0,1)}$QG also supports a class of motion integrals which neither form the kernel of the bracket nor are related to any explicit symmetries of the system via Noether's theorem.  To the best of our knowledge, such integrals have so far never been reported.  The model sustains the usual Rossby waves and a neutral model, both riding on a reference state, i.e., one with no currents, a special type of it being the vacuum state, i.e., one with vanishing energy. Positive definiteness of a general integral of motion, called a free energy, quadratic on the deviation from the reference state prevents perturbations of arbitrary shape and size on such a state from spuriously growing.  An exception is the vacuum state, as perturbations on this state have positive-semidefinite free energy associate.  Spurious growth of such perturbations cannot be prevented using free energy conservation.

The stability analysis carried out in the IL$^{(0,1)}$QG was applied on a family of basic states representing baroclinic zonal jets with curvature in the velocity vertical profile.  The IL$^{(0,1)}$QG has velocity independent of depth.  However, by the thermal-wind balance the velocity is implicitly vertically sheared.  We found that a current in this family is spectrally stable, i.e., with respect to infinitesimally small normal-mode perturbations, when the gradients of vertically averaged (mean) and layer thickness have like signs and the ratio of mean buoyancy and buoyancy frequency gradients are larger than a fraction of the reference layer thickness and vice versa.  In the limit of very weak stratification, the growth rate of long perturbations, i.e., with wavelengths of the order of the equivalent barotropic Rossby deformation radius, are found to saturate. By contrast, perturbations with wavelengths of the order of the gravest Rossby radius have vanishing growth rate.  Only a subset of the spectrally stable states was possible to be shown stable with respect to finite-size perturbations of arbitrary structure using the integrals of motion by means of the application of Arnold's method. These integrals exclude the weak Casimir integrals, which form the kernel of the Lie–Poisson bracket for the potential vorticity evolution independent of the details of the buoyancy as this is advected under the flow. The role of these integrals in constraining stratified thermal flow remains to be understood.  The states shown stable using Arnold's method are of Lyapunov type, i.e., the instantaneous perturbation distance to such states as measured using an $L^2$ norm is bounded at all times by a multiple of the initial distance to them.  These Lyapunov-stable states were used to \emph{a priori} bound the nonlinear growth of perturbations to spectrally unstable states, thereby preventing evolution in the IL$^{(0,1)}$QG from undergoing an ultraviolet explosion.  This is extensible to perturbations to any unstable state, including the vacuum state, whose free energy can be vanishing.

Our findings do not support the generality of earlier numerical evidence suggesting that suppression of sub-deformation scale (i.e., submesoscale) wave activity is a result of the inclusion of stratification in thermal shallow-water theory.  We backed up the latter with results of direct, fully nonlinear simulations of the IL$^{(0,1)}$QG and IL$^0$QG, which did not include topographic forcing as previously considered.

\subsection{Outlook}

Left for future investigation is uncovering the origins of the non-Casimir integrals of motion in the IL$^{(0,1)}$QG model. Preliminary investigations utilizing the primitive-equation parent model \cite{Beron-21-POFb} suggest that these conservation laws are linked through Noether's theorem, akin to Casimirs, to the invariance of Eulerian variables under particle relabeling. Also slated for future exploration is the integration of structure-preserving algorithms to provide a more precise framing of the fully nonlinear evolution in the IL$^{(0,1)}$QG model compared to the IL$^0$QG model. This endeavor should be approached in tandem with understanding turbulence dynamics in these models, which deviate from those in HLQG due to the absence of enstrophy conservation, and its implications for observations. One notable challenge is devising a geometric (Lie--Poisson) integrator for flows on the annulus, as current methods are limited to the torus or the sphere \cite[e.g.,][]{Modin-Milo-19, Cifani-etal-22}.

\section*{Glossary of fundamental geophysical fluid concepts}

\begin{description}
    \item[\textbf{Coriolis force.}] The only ``fictitious'' force required to make the fluid motion equations valid in a frame rotating with the Earth when the centrifugal force in that frame is balanced out by the gravitational force on a horizontal plane.
    \item[\textbf{Coriolis parameter.}] Twice the Earth's local vertical angular velocity.
    \item[\textbf{Boussinesq approximation.}] An approximation where density variations are ignored except where they appear multiplied by the acceleration due to gravity, that is, the vector composition of the gravitational and centrifugal accelerations.
    \item[\textbf{Primitive equations.}] Three-dimensional Euler equations for arbitrary stratified incompressible fluid with Coriolis force in the horizontal and the hydrostatic and Boussinesq approximations.  
    \item[\textbf{Rossby number.}] A dimensionless number defined as the ratio of inertial to Coriolis acceleration.
    \item[\textbf{Quasigeostrophic equations.}] Leading-order approximation to the primitive equations in the asymptotic limit of small Rossby number.
    \item[\textbf{Thermal-wind balance.}] A balance between the vertical derivative of the Coriolis force and the horizontal buoyancy gradient.
    \item[\textbf{Baroclinic instability.}] Instability resulting when a current is perturbed off from thermal-wind balance.
    \item[\textbf{Reduced gravity.}] An approximation suitable for describing upper-ocean variability, where an active fluid layer is bounded above by a rigid interface with the atmosphere and floats freely above a motionless, infinitely deep fluid. 
\end{description}

\section*{Acknowledgments}

\add[]{Constructive criticism from an anonymous reviewer resulted in enhancements to the paper.} Incorporating feedback from an anonymous reviewer on a prior draft of this paper contributed to its enhancement. We thank Philip Morrison for calling our attention to vacuum states.  We owe Darryl Holm our renewed interest in thermal fluids, originally introduced to us by the late Pedro Ripa while being graduate students of Departamento de Oceanograf\'ia F\'isica of CICESE (Ensenada, Baja California, Mexico) under his direction.

\appendix

\section*{Appendix: Numerical details of the direct simulations}\label{sec:app}

The numerical simulation of the IL$^{(0,1)}$QG employed a pseudosepectral scheme with fast Fourier transform (FFT) in $x$ and discrete sine transform (DST) in $y$ to invert \eqref{eq:inv}, but as written for the \emph{perturbation} to the (unstable) basic state. We consistently wrote the entire set \eqref{eq:IL01} accordingly. The DST imposes a homogeneous Dirichlet boundary condition on $\delta\bar\psi$ without spoiling the strict zero-normal-flow at the channel northern and southern walls boundary conditions, namely, $\partial_x\bar\psi\vert_{y=0,W} = 0$.  A total of 512 grid points in each direction, zonal and meridional, was considered.  Differentiation was done using dialiased FFT in each direction using the $\frac{3}{2}$ zero-padding rule (we tried Chebyshev differentiation in $y$, but this led to numerical instability).  The equations were forward advanced using a fourth-order Runge--Kutta method with time step $\Delta t U_\sigma/R \approx 0.0001$ as resulting by applying the Courant--Friedrichs--Lewy condition.  Finally, a small amount (roughly $-1.5\times 10^{-11}U_\sigma R^3$) of biharmonic hyperviscosity was included to stabilize the time step.  The same numerical treatment was applied to the simulation of the IL$^0$QG. Our confidence on the simulations is measured by how well the conservation laws in the IL$^{(0,1)}$QG and IL$^0$QG are represented.  Take energy ($\mathcal E$), momentum ($\mathcal M$), and two Casimirs common to the IL$^{(0,1)}$QG and IL$^0$QG, e.g., $\mathcal C_{1,0}$ and $\mathcal C_{1,\cos\psi_\sigma}$; cf.\ \eqref{eq:C-IL01} and \eqref{eq:C-IL0}.  The absolute magnitude of the relative error with respect to the values of these quantities initially are about 19.2, 2.6, 2.8, and 2.9\% for the IL$^{(0,1)}$QG, and 17.3, 12.2, 3.0, and 3.3\% for the IL$^0$QG 

\bibliographystyle{alpha}
%\bibliography{fot}
\newcommand{\etalchar}[1]{$^{#1}$}

\end{document}